%% file: paper.tex
\documentclass{lime}
\usepackage{shortcuts}
\pdfoutput=1

\title{Deterministic and Las Vegas Algorithms\\ for Sparse Nonnegative Convolution}
\author{Karl Bringmann}
\author{Nick Fischer}
\author{Vasileios Nakos}
\affil{Saarland University and Max Planck Institute for Informatics,\\Saarland Informatics Campus, Saarbrücken, Germany}

\begin{document}
\maketitle

\begin{abstract}
Computing the convolution $A \star B$ of two length-$n$ integer vectors $A, B$ is a core problem in several disciplines. It frequently comes up as a subroutine in various problem domains, e.g.\ in algorithms for Knapsack, $k$-SUM, All-Pairs Shortest Paths, and string pattern matching problems. For these applications it typically suffices to compute convolutions of \emph{nonnegative} vectors. This problem can be classically solved in time $O(n \log n)$ using the Fast Fourier Transform.

However, in many applications the involved vectors are \emph{sparse} and hence one could hope for \emph{output-sensitive} algorithms to compute nonnegative convolutions. This question was raised by Muthukrishnan and solved by Cole and Hariharan (STOC~'02) by a randomized algorithm running in near-linear time in the (unknown) output-size $t$. Chan and Lewenstein (STOC~'15) presented a deterministic algorithm with a $2^{\Order(\sqrt{\log t\cdot  \log\log n})}$ overhead in running time and the additional assumption that a small superset of the output is given; this assumption was later removed by Bringmann and Nakos~(ICALP~'21).

In this paper we present the first deterministic near-linear-time algorithm for computing sparse nonnegative convolutions. This immediately gives improved deterministic algorithms for the state-of-the-art of output-sensitive Subset Sum, block-mass pattern matching, $N$-fold Boolean convolution, and others, matching up to $\log$-factors the fastest known randomized algorithms for these problems. Our algorithm is a blend of algebraic and combinatorial ideas and techniques.

Additionally, we provide two fast \emph{Las Vegas} algorithms for computing sparse nonnegative convolutions. In particular, we present a simple  $O(t \log^2 t)$ time algorithm, which is an accessible alternative to Cole and Hariharan's algorithm. Subsequently, we further refine this new algorithm to run in Las Vegas time $O(t \log t \cdot \log \log t)$, which matches the running time of the dense case apart from the $\log \log t$ factor.
\end{abstract}

\begin{funding}
This work is part of the project TIPEA that has received funding from the European Research Council (ERC) under the European Unions Horizon 2020 research and innovation programme (grant agreement No.~850979).
\end{funding}

\input{sections/introduction}
\input{sections/preliminaries}

\input{sections/deterministic}
\input{sections/las-vegas}

\bibliography{refs}

\appendix
\input{sections/linear-hashing}

\end{document}

%% file: sections/introduction.tex
% !TEX root = ../paper.tex
\section{Introduction} \label{sec:introduction}
The \emph{convolution} of two integer vectors $A, B$ is the vector $A \conv B$ which is defined coordinate-wise by~$(A \conv B)_k = \sum_{i+j=k} A_i \cdot B_j$. Computing convolutions of integer vectors $A, B$ is a fundamental computational primitive, which arises in several disciplines of science and engineering. It has been a vital component in fields like signal processing, deep learning (convolutional neural networks) and computer vision. Inside traditional algorithm design it is crucially used as a subroutine in $k$-SUM~\cite{ChanL15}, Subset Sum~\cite{Bringmann17,KoiliarisX19,BringmannN20,BringmannN21b} and various string problems~\cite{FischerP74,Indyk98,ColeH02}, to name a few.

The aforementioned applications of interest within theoretical computer science typically come in the form of \emph{nonnegative convolution}, where the vectors $A, B$ have nonnegative entries. In fact, for many applications it suffices to solve the simpler \emph{Boolean convolution} problem---here, the vectors $A, B$ have $0$--$1$ entries and the task is to compute the vector~$A \ostar B$ with entries $(A \ostar B)_k = \bigvee_{i+j=k} A_i \land B_j$. This problem is equivalent to the computation of sumsets~$X+ Y = \{x + y : x \in X,\allowbreak y \in Y\}$ and this interpretation shows up very often in $k$-SUM and Subset Sum algorithms.

Moreover, Boolean (or nonnegative) convolutions form an essential ingredient to the \emph{partition-and-convolve} design paradigm. The typical task for a problem approachable by a partition-and-convolve algorithm is to check for solutions of \emph{all prescribed sizes $k$}. The general idea is to partition the search space into (usually) two parts, each of which is solved recursively. In that way, a size-$k$ solution to the original problem is split into two parts of sizes $i, j$ such that \mbox{$i + j = k$}. Therefore, to check whether there exists a size-$k$ solution to the original problem we recombine the recursive computations by a Boolean convolution. This approach is very flexible and can also be applied to other convolution-type problems; for instance, using nonnegative convolutions in place of Boolean convolutions corresponds to \emph{counting} solutions of all prescribed sizes. 

Classically, the convolution of length-$n$ vectors can be computed in deterministic time $\Order(n \log n)$ using the Fast Fourier Transform (FFT). It is widely conjectured that this algorithm is optimal  but the evidence is scarce~\cite{Ailon13,AfshaniFKL19} and this remains an important open problem. Furthermore, it is known that nonnegative convolution, general convolution (where entries can also be negative or complex) and the computation of Discrete Fourier transforms (DFT) are computationally equivalent, as each one can be reduced to the other.\footnote{For dense vectors, the nonnegativity assumption can be removed by appropriately increasing all entries. For the equivalence of computing convolutions and DFTs we remark that it is standard to express convolutions using DFT and inverse DFT, and the reverse direction is known as well~\cite{Bluestein70} (assuming complex exponentials can be evaluated in constant time).}

The situation is less clear for \emph{sparse convolutions} in all regards. Here the goal is to achieve \emph{output-sensitive} algorithms where we analyze the running time in terms of $t$, the combined number of nonzero entries in $A$, $B$ and $A \conv B$ (in the Boolean and nonnegative cases, $t$ is dominated by the number of nonzero entries in $A \conv B$). The need for such a primitive appears in many situations, e.g.~\cite{ColeH02,HassaniehAKI12,ChanL15,BringmannN20,BringmannN21}, as one may often be interested in algorithms that run in time proportional to the actual complexity of the output rather the space the output points live in. These type of problems have been investigated by different communities, including fine-grained complexity~\cite{ChanL15} and string algorithms~\cite{ColeH02}, computer algebra~\cite{Roche18} and compressed sensing~\cite{HassaniehIKP12,FoucartR13}, and they are very closely related to the famous sparse recovery problem, see e.g.~\cite{GilbertI10,FoucartR13}. Subsequently, we review the most relevant literature on sparse convolutions.

\paragraph*{Randomized Algorithms for Sparse Convolution}
A large body of work addresses this problem~\cite{Muthukrishnan95,ColeH02,Roche08,MonaganP09,VanDerHoevenL12,ArnoldR15,ChanL15,Roche18,Nakos20,GiorgiGC20,BringmannFN21}. The first breakthrough was a randomized Las Vegas algorithm for sparse nonnegative convolution in time $\Order(t \log^2 n)$, obtained by Cole and Hariharan~\cite{ColeH02}. Subsequent work improved upon this result in two directions: On the one hand, the nonnegativity assumption can be removed by a Monte Carlo algorithm with the same running time $\Order(t \log^2 n)$~\cite{Nakos20}, or with bit-complexity $\widetilde\Order(t \log n)$~\cite{GiorgiGC20}. On the other hand, there exists an improved Monte Carlo algorithm for sparse nonnegative convolution in time $\Order(t \log t + \polylog n)$~\cite{BringmannFN21}, which is optimal assuming that the dense problem requires FFT time $\Theta(n \log n)$. Most of these algorithms rely on a hashing-based approach.

Another more algebraic avenue to sparse convolution algorithms is via polynomial evaluation and interpolation. At the heart of this approach lies an old algorithm called \emph{Prony's method}~\cite{Prony1795} which allows to efficiently interpolate a sparse polynomial. Unfortunately, this algorithm involves heavy algebraic computations (see~\cite{Roche18} for a detailed survey) and the currently only known way to achieve near-linear running time in $t$ uses randomization~\cite{Roche18}. Thus, near-linear-time sparse convolution algorithms resulting from this approach are randomized as well. We compare our work against Prony's method in more detail in \cref{sec:tech}.

\bigskip\noindent
Note that in many applications of interest a sparse convolution algorithm is called many times with varying output size and thus it is desirable to have deterministic (or Las Vegas) algorithms for performing such a task. Additionally, the fastest known algorithm in the dense case (FFT) is deterministic, and thus it is natural to wonder to what extent randomization is necessary in the sparse case.

\paragraph*{Deterministic Algorithms for Sparse Convolution}
The first nontrivial deterministic result for sparse nonnegative convolution is a data structure that, after preprocessing one of the vectors in time $\Theta(t^2)$, computes the convolution with any given query vector in near-linear time $\Order(t \log^3 t)$~\cite{AmirKP07}. Later, Chan and Lewenstein~\cite{ChanL15} devised a deterministic algorithm running in time $t \cdot 2^{\Order(\sqrt{\log t \log\log n})}$, without preprocessing. Their algorithm is limited in the sense that it expects as an additional input the support (i.e., the set of nonzero coordinates) of $A \conv B$. This assumption can be removed as shown by Bringmann and Nakos~\cite{BringmannN21}; in \cref{sec:tech} we provide more details. In summary, the state-of-the-art deterministic algorithms for computing sparse nonnegative convolutions either require heavy precomputations or fail to achieve near-linear time $\Order(t \polylog n)$. Our driving question is therefore:
\begin{center}
    \medskip
    \emph{Can sparse nonnegative convolutions be computed\\in deterministic time $\Order(t \polylog n)$?}
\end{center}

\subsection{Our Results} \label{sec:results}
Our main result is an affirmative answer to our driving question.

\begin{theorem}[Deterministic] \label{thm:deterministic}
There is a deterministic algorithm to compute the convolution of two nonnegative vectors $A, B \in \Nat^n$ in time $\Order(t \polylog(n\Delta))$, where $t = \norm{A \star B}_0$ denotes the number of nonzero entries in $A \conv B$ and $\Delta = \norm{A \star B}_\infty$ is the maximum entry size.
\end{theorem}

\noindent
This result improves the previously best known time $t \cdot 2^{\Order(\sqrt{\log t \cdot \log\log n})}$ obtained by~\cite{ChanL15,BringmannN21}.

As a corollary we can efficiently derandomize known algorithms for several problems which use sparse nonnegative convolution as a subroutine. For all these applications, we can simply replace the former randomized algorithms with our deterministic one in a black-box manner. Of course, for the same derandomization we could alternatively use the \mbox{$t \cdot 2^{\Order(\sqrt{\log t \log\log n})} = t \cdot n^{\order(1)}$}-time algorithm from~\cite{ChanL15,BringmannN21} and therefore our contribution can alternatively be seen as improving the best deterministic time from $T^{1+\order(1)}$ to $\widetilde\Order(T)$. Specifically, we obtain improvements for the following problems:
\begin{itemize}
\itemdesc{Output-Sensitive Subset Sum:} Given a set $X$ of integers and a threshold $\tau$, compute the set $S$ of all numbers less than $\tau$ which can be expressed as a subset sum of $X$. The best-known randomized algorithm runs in time $\widetilde\Order(|S|^{4/3})$~\cite{BringmannN20}, and it can be derandomized in same running time.
\itemdesc{$N$-fold Boolean Convolution:} Given $N$ Boolean vectors $A_1, \dots, A_N$, compute the Boolean convolution $A_1 \ostar \dots \ostar A_N$ (with or without wrap-around) in input- plus output-sensitive time. It was recently shown that this problem can be solved in randomized near-linear time $\Order(t \polylog n)$~\cite{BringmannN21}. Our derandomization achieves the same running time. This yields a new deterministic near-linear-time algorithm for Modular Subset Sum which is rather different than the known ones~\cite{AxiotisABJNTW21}, as discussed in~\cite{BringmannN21}.
\itemdesc{Block-Mass Pattern Matching:} Given a length-$n$ text $T$ and a length-$m$ pattern $P$ over the alphabet $\Nat$, the task is to output all possible indices $0 \leq k_0 \leq \dots \leq k_m \leq n$ such that \raisebox{0pt}[0pt][0pt]{$P_i = \sum_{k_i \leq j < k_{i+1}} T_j$} for all positions $i \in [m]$. Building on the data structure from~\cite{AmirKP07}, this problem is known to be solvable in deterministic time $\widetilde\Order(n + m)$ after preprocessing the text in time $\Order(n^2)$~\cite{AmirBP14}. The preprocessing time was later reduced to~$\Order(n^{1+\epsilon})$, for any $\varepsilon > 0$~\cite{ChanL15}. We entirely remove the necessity to precompute and thereby reduce the total running time to $\widetilde\Order(n + m)$.
\itemdesc{3SUM in Special Cases:} In a breakthrough paper, Chan and Lewenstein~\cite{ChanL15} use sophisticated techniques to obtain randomized and deterministic subquadratic algorithms for a variety of problems related to $3$-SUM, such as bounded monotone two-dimensional 3SUM, bounded monotone $(\min,+)$-convolution, clustered integer $3$-SUM, etc. The precise running time of their deterministic algorithm for these problems is $\Order(n^{1.864})$. Here we remove an~$\order(1)$ overhead in the exponent which is invisible due to rounding the constant in the exponent.
\end{itemize}

\noindent
In addition to our new deterministic algorithm, we improve the state-of-the-art \emph{Las Vegas} algorithms for sparse nonnegative convolutions in two regards: \emph{simplicity} and \emph{efficiency}. In fact, to the best of our knowledge the only known Las Vegas algorithm is due to Cole and Hariharan~\cite{ColeH02}; all randomized algorithms published later have only Monte Carlo guarantees~\cite{Roche08,ArnoldR15,Roche18,Nakos20,GiorgiGC20,BringmannFN21}. The expected running time of~\cite{ColeH02} is $\Order(t \log^2 n)$ and moreover, they prove the additional guarantee that their algorithm terminates in time $\Order(t \log^2 n)$ with high probability $1 - \frac1n$. However, the algorithm is very complicated, involves various string problems as subtasks and the precomputation of a large prime number. We provide an accessible alternative with the same theoretical guarantees; the simplest version can be summarized in 13 lines of pseudocode (\cref{alg:las-vegas} on \cpageref{alg:las-vegas}).

\begin{theorem}[Simple Las Vegas] \label{thm:las-vegas}
Given nonnegative vectors $A, B \in \Nat^n$, there exists an algorithm to compute their convolution $A \conv B$ in expected time $\Order(t \log^2 t)$, where $t = \norm{A \star B}_0$. Moreover, with probability $1 - \delta$ the running time is bounded by $\Order(t \log^2(t / \delta))$.
\end{theorem}

\noindent
In comparison to Cole and Hariharan's algorithm, our algorithm runs slightly faster in expectation (at least if~$t \ll n$) and achieves the same high-probability guarantee (indeed, by setting $\delta = \frac1n$ the running time is bounded by $\Order(t \log^2 n)$ with probability at least $1 - \frac1n$). We further show how to reduce the expected running time, achieving optimality up to a $\log \log$ factor.

\begin{theorem}[Fast Las Vegas] \label{thm:fast-las-vegas}
Given nonnegative vectors $A, B \in \Nat^n$, there exists an algorithm to compute their convolution $A \conv B$ in expected time $\Order(t \log t \log\log t)$, where $t = \norm{A \conv B}_0$.
\end{theorem}

\noindent
Assuming that FFT-time $\Order(n \log n)$ is best-possible for computing dense convolutions, the best-possible algorithm for computing sparse convolutions requires time $\Omega(t \log t)$. Hence, our algorithm is likely optimal, up to the $\log\log t$ factor.

\subsection{Technical Overview} \label{sec:tech}
In this section we briefly outline the ideas behind \cref{thm:deterministic,thm:las-vegas,thm:fast-las-vegas}.

\subsubsection{Deterministic Algorithm}
The key machinery powering our deterministic algorithm (\cref{thm:deterministic}) is a basic result from structured linear algebra which can be viewed as efficiently evaluating and interpolating \emph{sparse} polynomials---under certain conditions. We present the algorithm by first explaining that key part (Part~1) and the assumptions it requires. In Parts~2 and~3 we then remove these assumptions by appropriate precomputations.

\paragraph*{Part 1: Evaluation \& Interpolation}
The high-level approach follows the typical evaluation and interpolation pattern. Any vector $V$ can be viewed as a polynomial $V(X) = \sum_{i=0}^{n-1} V_i X^i$. In that analogy, computing the convolution $A \conv B$ of two vectors $A, B$ corresponds to computing the product of their respective polynomials $A(X) \cdot B(X)$. The idea is to:
\begin{enumerate}
\item Evaluate $A(X)$ and $B(X)$ at some carefully chosen points $\omega^0, \omega^1, \dots, \omega^{t-1}$,
\item Compute the product $A(\omega^i) \cdot B(\omega^i)$ for all $i = 0, \dots, t-1$,
\item Interpolate the (hopefully unique) polynomial $C(X)$ with evaluations $C(\omega^i) = A(\omega^i) \cdot B(\omega^i)$.
\end{enumerate}
In this way, we have reduced the task to the evaluation and interpolation of sparse polynomials. Let us start with the evaluation problem: Note that computing $V(\omega^0), \dots, V(\omega^{t-1})$ is equivalent to computing the following matrix-vector product:
\begin{equation}
    \begin{bmatrix}
        V(\omega^0) \\ V(\omega^1) \\ \vdots \\ V(\omega^{t-1})
    \end{bmatrix}
    =
    \begin{bmatrix}
        1 & 1 & \cdots & 1 \\
        \omega^{x_1} & \omega^{x_2} & \cdots & \omega^{x_t} \\
        \vdots & \vdots & \ddots & \vdots \\
        \omega^{x_1 (t-1)} & \omega^{x_2 (t-1)} & \cdots & \omega^{x_t (t-1)}
    \end{bmatrix}
    \begin{bmatrix}
        V_{x_1} \\ V_{x_2} \\ \vdots \\ V_{x_t}
    \end{bmatrix}. \tag{\textasteriskcentered}\label{eq:vandermonde}
\end{equation}
This matrix has a very special form: It is the transpose of a Vandermonde matrix. It is known that performing linear algebra operations (such as computing matrix-vector products, or solving linear systems) with transposed Vandermonde matrices can be implemented in $\Order(t \log^2 t)$ field operations~\cite{KaltofenL88,Li00,Pan01}. Thus, by viewing the integer vector $V$ as a vector over some appropriately large finite field, we can evaluate $V(\omega^0), \dots, V(\omega^{t-1})$ in near-linear time. (We remark that this algorithm is numerically unstable, so using complex arithmetic is not an option.)

To perform the inverse task of interpolating the coefficients $V_{x_1}, \dots, V_{x_t}$ given the evaluations $V(\omega^0), \dots, V(\omega^{t-1})$, we view \eqref{eq:vandermonde} as a system of linear equations with indeterminates $V_{x_1}, \dots, V_{x_t}$. As mentioned before this problem can be also solved in time $\Order(t \log^2 t)$. This nearly yields the algorithm, however, there are two major obstacles. First, in order to obtain a unique solution, the equation system should be nonsingular. It is easy to see that this is equivalent to the condition that~$\omega^{x_1}, \dots, \omega^{x_t}$ are pairwise distinct. A reasonable way to achieve this is to let~$\omega$ be a finite field element with multiplicative order at least $n \geq \deg(V)$. In Part~2 we explain how to obtain such an element. Second, in order to write down the equation system we have to know the indices~$x_1, \dots, x_t$, i.e., the support of $V$. Concretely, in the algorithm we call the sparse interpolation problem for $V = A \conv B$ and we therefore need to know $\supp(A \conv B)$ in advance. In Part~3 we discuss a recursive ``scaling trick'' to precompute a small superset of $\supp(A \conv B)$.

\paragraph*{Part 2: Finding Large-Order Elements}
In this part we care about finding an element~$\omega$ with multiplicative order at least $n$ in a finite field of size $\gg n$. There is a simple randomized algorithm: Pick a random element. Unfortunately, the best-known \emph{deterministic} algorithms for finding a large-order element in a given prime field $\Field_p$ require time polynomial in $p$~\cite{Cheng05}. Thus it seems intractable to work over a finite field $\Field_p$ with $p \geq n$ as originally intended.

Fortunately, in a finite field~$\Field_q = \Field_{p^m}$ with prime power order, it is possible to find large-order elements in time $\poly(p, m)$~\cite{Cheng07,Shoup90,Shparlinski92}. Specifically, setting $p, m = \polylog(n)$ we can find an element $\omega$ with order at least $n$ in time $\polylog(n)$~\cite{Cheng07}. Working over a finite field with small characteristic $p \leq \polylog(n)$ has another drawback though: We cannot recover the entries of the vector $A \conv B$ (which can have size up to $n$, even if $A$ and $B$ are bit-vectors to begin with). We remedy this problem by computing the convolution $A \conv B$ over \emph{several} finite fields $\Field_{q_1}, \Field_{q_2}, \dots$, and use the Chinese Remainder Theorem to identify the correct integer solution afterwards.

\paragraph*{Part 3: Recursively Computing the Support}
We finally discuss how to precompute the support $\supp(A \conv B)$. In fact, it suffices to compute a superset $T \supseteq \supp(A \conv B)$ with small size $|T| \leq \Order(t)$. We exploit a trick which was first applied to the context of convolutions in~\cite{BringmannN21}; see also~\cite{BringmannN20,BringmannFHSW20}. Construct smaller vectors $A', B'$ of length~$\frac n2$ by $A'_i = A_i + A_{i+n/2}$ and $B'_j = B_j + B_{j+n/2}$ (that is, we fold $A, B$ in half). We can recursively compute the convolution $C' = A' \conv B'$. Then we extract $T$ as
\begin{equation*}
    T = \big\{\, k, k + \tfrac n2, k + n : k \in \supp(C') \,\big\}.
\end{equation*}
This choice is correct: Clearly $|T| \leq 3t$, and it is easy to verify that $T$ is indeed a superset of $\supp(A \conv B)$. The recursion only reaches depth $\log n$, and thus incurs a logarithmic overhead in the running time.

\bigskip\noindent
By combining Parts~2 and~3 we overcome both obstacles outlined in Part~1, and solve sparse nonnegative convolution in deterministic time $\Order(t \polylog(n \Delta))$.

\paragraph*{Comparison to Prony's Method}
Note that our main contribution can also be viewed as a deterministic near-linear-time algorithm to interpolate a univariate sparse polynomial \emph{with nonnegative coefficients}. The classical approach to the sparse interpolation problem is by Prony's method---an old algorithm first discovered by Prony in 1795~\cite{Prony1795}, and rediscovered later by Ben-Or and Tiwari~\cite{BenOrT88}; see~\cite{Roche18} for a detailed survey. Prony's method involves heavy algebraic computations such as finding the minimal solution to a linear recurrence, polynomial root finding, computing discrete logarithms and linear algebra with (transposed) Vandermonde systems. These computations can be carried out in near-linear time, but only using randomization~\cite{Kaltofen10,Roche18}. Our algorithm is similar to Prony's method with two essential modifications: We replace the computationally expensive parts using the combinatorial trick (Part~3, here we critically use that the vectors are nonnegative) and derandomize the remaining steps (Parts~1 and~2) using classical methods.

\subsubsection{Las Vegas Algorithms}
\begin{algorithm}[t]
\caption{} \label{alg:las-vegas}
\begin{algorithmic}[1]
\Input{Nonnegative vectors $A, B \in \Nat^n$}
\Output{$C = A \conv B$}
\smallskip
\For{$m \gets 1, 2, 4, \dots, \infty$} \label{alg:las-vegas:line:guess}
    \RepeatTimes{$2\log m$} \label{alg:las-vegas:line:repeat}
        \State Sample a linear hash function $h : [n] \to [m]$ \label{alg:las-vegas:line:hash}
        \State Compute $X \gets h(A) \conv_m h(B)$
        \State Compute $Y \gets h(\partial A) \conv_m h(B) + h(A) \conv_m h(\partial B)$
        \State Compute $Z \gets h(\partial^2 A) \conv_m h(B) + 2 h(\partial A) \conv_m h(\partial B) + h(A) \conv_m h(\partial^2 B)$
        \State Initialize $R \gets (0, \dots, 0)$
        \ForEach{$k \in [m]$} \label{alg:las-vegas:line:buckets}
            \If{$X_k \neq 0$ \AND{} $Y_k^2 = X_k \cdot Z_k$} \label{alg:las-vegas:line:condition}
                \State $z \gets Y_k / X_k$ \label{alg:las-vegas:line:position}
                \State $R_z \gets R_z + X_k$ \label{alg:las-vegas:line:update}
            \EndIf
        \EndForEach
    \EndRepeatTimes
    \State Let $C$ be the coordinate-wise maximum of all vectors $R$ \label{alg:las-vegas:line:max}
    \If{$\norm C_1 = \norm A_1 \cdot \norm B_1$} \Return $C$ \label{alg:las-vegas:line:return}
    \EndIf
\EndFor
\end{algorithmic}
\end{algorithm}

Next, we outline the idea behind proving \cref{thm:las-vegas}. \cref{alg:las-vegas} is a simple Las Vegas algorithm with expected running time $\Order(t \log^2 t)$ as claimed in \cref{thm:las-vegas}; however, to obtain the tail bound on the running time one has to slightly refine \cref{alg:las-vegas}. We provide this refinement along with a detailed analysis in \cref{sec:las-vegas}; for the rest of the overview we will analyze the simple version in \cref{alg:las-vegas}.

To understand the pseudocode, we first clarify some notation: For a vector $A$, we denote by $\partial A$ its \emph{derivative} defined coordinate-wise as $(\partial A)_i = i \cdot A_i$. More generally, we denote by $\partial^d A$ its \emph{$d$-th derivative} with $(\partial^d A)_i = i^d \cdot A_i$. This definition is in slight dissonance with the analogous definition for polynomials (which would require the derivative vector to be scaled and \emph{shifted}), but we prefer this version as it leads to a slightly simpler algorithm.

Moreover, we define hashing for vectors: For a hash function $h : [n] \to [m]$ and a length\=/$n$ vector $A$, define the length-$m$ vector $h(A)$ via $h(A)_j = \sum_{i : h(i) = j} A_i$. The operator~$\conv_m$ denotes convolution with wrap-around (see \cref{sec:preliminaries} for details).

Let us outline the high-level idea of \cref{alg:las-vegas}. The outer loop (\cref{alg:las-vegas:line:guess}) guesses the correct sparsity, i.e., as soon as the outer loop reaches a value $m \geq \Omega(t)$ we expect the algorithm to terminate. Each iteration of the repeat-loop (\cref{alg:las-vegas:line:repeat}) is supposed to produce a vector $R$ which closely approximates $A \conv B$. More specifically, we prove that $R$ satisfies the following two properties:
\begin{enumerate}
\item It always holds that $R \leq A \conv B$ (coordinate-wise).
\item Equality is achieved at any coordinate with constant probability (provided that the outer loop has reached a sufficiently large value $m \geq \Omega(t)$).
\end{enumerate}
It follows that $C$, the coordinate-wise maximum of several vectors $R$, also always satisfies $C \leq A \conv B$. Hence, the algorithm never outputs an incorrect solution. Indeed, since $C$ and $A \conv B$ are nonnegative vectors, the vector $C = A \conv B$ is the only one simultaneously satisfying $C \leq A \conv B$ and $\norm C_1 = \norm{A \conv B}_1 = \norm A_1 \cdot \norm B_1$. To see that \cref{alg:las-vegas} terminates fast, note that the repeat-loop runs for $\Omega(\log m)$ iterations and thus, using the second claim we correctly assign \emph{all} coordinates with high probability.

The crucial part is to prove that $R$ satisfies the claims~1 and~2. Intuitively, $R$ consists of all nonzero entries from $A \conv B$ which did not suffer from a collision with another nonzero entry. For a more formal argument, we analyze the inner-most loop (\cref{alg:las-vegas:line:buckets}). For starters, focus on an iteration $k \in [m]$ and suppose that there is only a single nonzero entry in $A \conv B$, say at $z$, which is hashed to the bucket~$k$.\footnote{Strictly speaking, that condition is not sufficient because linear hashing is only ``almost'' additive. We ignore this technical issue in the overview and give the full analysis in \cref{sec:las-vegas}.} In this case we have $X_k = (A \conv B)_z$, $Y_k = z \cdot (A \conv B)_z$ and $Z_k = z^2 \cdot (A \conv B)_z$. As a consequence, the conditions ``$X_k \neq 0$'' and ``$Y_k^2 = X_k \cdot Z_k$'' in \cref{alg:las-vegas:line:condition} are satisfied. The algorithm then correctly identifies~$z$ in \cref{alg:las-vegas:line:position} and updates ``$R_z \gets R_z + (A \conv B)_z$'' as intended.

However, to prove claim~1 (which is ultimately responsible for the Las Vegas guarantee), we have to be certain that \cref{alg:las-vegas:line:position,alg:las-vegas:line:update} are only executed if there is a single entry hashed to the $k$-th bucket (otherwise, the index $z$ computed in \cref{alg:las-vegas:line:position} is likely to be nonsense). The key insight is that the simple test ``$Y_k^2 = X_k \cdot Z_k$'' in \cref{alg:las-vegas:line:condition} suffices, as can be proven by the following lemma (see \cref{sec:las-vegas} for a proof). 

\begin{restatable}[Testing $1$-Sparsity]{lemma}{lemmaonesparsity} \label{lem:1-sparse}
If $V$ be a nonnegative vector, then $\norm{\partial V}_1^2 \leq \norm V_1 \cdot \norm{\partial^2 V}_1$. This inequality is tight if and only if $\norm V_0 \leq 1$.
\end{restatable}

\noindent
This new tester is one of the reasons why we can achieve the claimed Las Vegas running time simplifying (and slightly improving) upon Cole and Hariharan's algorithm. This concludes the overview of our simple Las Vegas algorithm (\cref{thm:las-vegas}).

The insight behind our accelerated Las Vegas algorithm (\cref{thm:fast-las-vegas}) is that \cref{alg:las-vegas} already reaches a very good approximation after much less than $\Order(\log m)$ iterations of the inner loop. Indeed, after only $\Order(\log\log n)$ iterations we expect that algorithm has already recovered $A \conv B$ correctly up to a $(\log n)^{-\Omega(1)}$ fraction of the entries. At this point it becomes more efficient to switch to another recovery approach which exploits that $A \conv B - C$ is already quite sparse, as in~\cite{BringmannFN21}. In particular, since $A \conv B - C$ is a \emph{nonnegative} vector and its sparsity is at most $t / \log n$, say, we can use the hash function $h(x) = x \bmod p$ for~$p$ being a random prime in~$[t,2t]$. This family of hash functions (1) satisfies that $h(A) \conv_m h(B) - h(C) = h(A \conv B - C)$ (and thus preserves all cancellations) and (2) isolates a constant fraction of elements in $A \conv B - C$ with constant probability to clear up the rest of the elements. Note that it is important that $A \conv B-C$ is $t/\log n$ sparse instead of $t$ sparse for (2) to hold, because $h$ is only $\Order(\log n)$-universal. Choosing $O( \log t)$ different random primes and using the $1$-sparsity testing we arrive at our desired algorithm. For the sparsity test we require that the vector~$A \conv B - C$ is nonnegative.

One catch is that this approach only gives a $O(t \log t \cdot \log \log n)$-time algorithm (instead of the desired time with $\log\log t$ in place of $\log\log n$) due to the fact that $h(x)$ is $O(\log n)$-universal and hence the random prime must be chosen in an interval that is also dependent on $n$ rather than solely on $t$. To address this issue we apply the following precomputation: We hash to a $\poly(t)$-size universe and verify that this hashing was successful in Las Vegas randomized time, again using our $1$-sparsity tester. The details of this step appear in \cref{sec:las-vegas:sec:length-reduction}.

\subsection{Discussion and Open Problems}
Our work raises several questions.

\paragraph*{Better Deterministic Algorithms?}
By a closer inspection of the time analysis, our deterministic algorithm computes the convolution of sparse nonnegative vectors in time $\Order(t \log^5 (n\Delta) \polyloglog(n\Delta))$.
\begin{enumerate}
\item Can the running time be improved? In particular, is it possible to reduce the number of log factors or can we omit the dependence on $n$ or $\Delta$?
\item Can the restriction to nonnegative vectors be removed, or equivalently, is it possible to achieve \emph{sparse polynomial multiplication in deterministic near-linear time?} In our algorithm the only step which exploits nonnegativity is the recursive support computation (Part~3).
\end{enumerate}

\paragraph*{Better Las Vegas Algorithms?}
We proved that sparse nonnegative convolution is in Las Vegas time $\Order(t \log t \log\log t)$.

\begin{enumerate}[resume]
\item Is the restriction to nonnegative vectors necessary? This seems like a difficult question because we are not aware of algorithms that run in even slightly subquadratic time in $t$ (without using heavy pre-computation).
\item Can one achieve $\Order(t \log t)$ Las Vegas running time matching the running time of the dense case? More specifically, can the sparsity-testing technique which lead to our Las Vegas algorithms be extended and incorporated to obtain the optimal running time? Recall that the key step in the analysis is the application of \cref{lem:1-sparse}. This lemma can be generalized as follows: A nonnegative vector $V$ is at most $s$-sparse if and only if the following positive-semidefinite matrix is nonsingular:
\begin{equation*}
    \begin{bmatrix}
        \norm{\partial^0 V}_1 & \norm{\partial^1 V}_1 & \cdots & \norm{\partial^s V}_1 \\
        \norm{\partial^1 V}_1 & \norm{\partial^2 V}_1 & \cdots & \norm{\partial^{s+1} V}_1 \\
        \vdots & \vdots & \ddots & \vdots \\
        \norm{\partial^s V}_1 & \norm{\partial^{s+1} V}_1 & \cdots & \norm{\partial^{2s} V}_1 \\
    \end{bmatrix};
\end{equation*}
see for instance~\cite[Theorem 3A]{Lindsay89} for a proof. One approach for an improved Las Vegas algorithm would be to hash to $t/ \log t$ buckets using a linear hash function, recover each bucket as in~\cite{BringmannFN21} in $\Order(t \log t)$ time and, using the generalized sparsity-testing technique, verify that most buckets indeed have sparsity $\Order(\log t)$, which in turn means that all but a $1/\log t$-fraction of $A \conv B$ has been successfully recovered; then one can continue and recover the rest with $h(x) = x \bmod p$. Although promising, this approach suffers from precision issues (when implementing the $\Order(\log t)$-tester the numbers get too large) and hence does not lead to the desired $O(t \log t)$ time. It would be very interesting to find a way to circumvent this obstacle and obtain the ideal $O(t \log t)$ Las Vegas running time.
\end{enumerate}

%% file: sections/preliminaries.tex
% !TEX root = ../paper.tex
\section{Preliminaries} \label{sec:preliminaries}

\paragraph*{Machine Model}
Throughout this paper we work over the Word RAM model. In particular, logical and arithmetic operations on machine words take constant time. Concerning the sparse convolution problem, we assume that both the indices and entries of the given vectors fit into a constant number of machine words.

\paragraph*{Notation}
Let $\Int$ and $\Nat$ denote the integers and nonnegative integers, respectively. For a prime power~$q$, let $\Field_q$ denote the finite field with $q$ elements. We set $[n] = \{0, \dots, n-1\}$. We write $\poly(n) = n^{\Order(1)}$, $\polylog(n) = (\log n)^{\Order(1)}$ and $\polyloglog(n) = (\log \log n)^{\Order(1)}$.

We mostly denote vectors by $A, B, C$ with $A_i$ referring to the $i$-th coordinate in $A$. We define the \emph{convolution} of two length-$n$ vectors $A$ and $B$ as the vector $A \conv B$ of length $2n-1$ with
\begin{equation*}
    (A \conv B)_k = \sum_{\substack{i, j \in [n]\\i+j=k}} A_i \cdot B_j.
\end{equation*}
The \emph{cyclic convolution} $A \conv_m B$ is the length-$m$ vector with
\begin{equation*}
    (A \conv_m B)_k = \sum_{\substack{i, j \in [n]\\i+j \equiv k \mod m}} A_i \cdot B_j.
\end{equation*}
We refer to $\supp(A) = \{ i \in [n] : A_i \neq 0\}$ as the \emph{support} of $A$, we set $\norm A_0 = |\supp(A)|$ and say that $A$ is \emph{$s$-sparse} if $\norm A_0 \leq s$. If~$A$ is a vector with real entries, then we also define $\norm A_1 = \sum_i |A_i|$ and $\norm A_\infty = \max_i |A_i|$ in the usual way, and we say that $A$ is \emph{nonnegative} if all of its entries are nonnegative. We often hash length-$n$ vectors $A$ using an arbitrary hash function $h : [n] \to [m]$ to shorter length-$m$ vectors~$h(A)$ defined by
\begin{equation*}
    h(A)_j = \sum_{\substack{i \in [n]\\h(i) = j}} A_i.
\end{equation*}

\paragraph*{Finite Field Arithmetic}
Let $q = p^m$ be a prime power. Recall that the prime field $\Field_p$ can be represented as $\Int / p \Int$, the integers modulo~$p$. The field $\Field_q$ can be represented as $\Field_p[X] / \ideal{f}$ where $f \in \Field_p[X]$ is an arbitrary irreducible degree-$m$ polynomial. There is a deterministic algorithm to precompute such an irreducible polynomial $f \in \Field_p$ in time $\poly(p, m)$~\cite{Shoup88}; we will point out this step in our algorithms. Having precomputed $f$, we can perform the basic field operations in~$\Field_q$ using polynomial arithmetic in time $\widetilde\Order(\log q)$~\cite{vonzurGathenG13}.

Let us quickly recall some definitions from field theory. The \emph{multiplicative order} of an element $x$ is the smallest positive integer $i$ such that $x^i = 1$; we also call $x$ an \emph{$i$-th root of unity}. The \emph{minimal polynomial} of a field element $x \in \Field$ is defined as the smallest-degree monic polynomial (i.e., with leading coefficient $1$) over $\Field$ which vanishes at~$x$. We say that two field elements $x, y$ are \emph{conjugate} if their minimal polynomials coincide.

%% file: sections/deterministic.tex
% !TEX root = ../paper.tex
\section{Deterministic Algorithm} \label{sec:deterministic}
In this section we prove \cref{thm:deterministic}. We proceed in three steps, as outlined before.

\subsection{The Key Step: Evaluation \& Interpolation} \label{sec:deterministic:sec:interpolation}
The main algebraic ingredient to the algorithm is the following result about efficient computations with transposed Vandermonde matrices. For a proof, see e.g.~\cite{KaltofenL88,Li00,Pan01}.

\begin{theorem}[Transposed Vandermonde Systems] \label{thm:vandermonde}
Let $\Field$ be a field. Given pairwise distinct elements $a_0, \dots, a_{n-1} \in \Field$ and a vector $x \in \Field^n$, let
\begin{equation*}
    M = \begin{bmatrix} 1 & 1 & \cdots & 1 \\ a_0 & a_1 & \cdots & a_{n-1} \\ a_0^2 & a_1^2 & \cdots & a_{n-1}^2 \\ \vdots & \vdots & \ddots & \vdots \\ a_0^{n-1} & a_1^{n-1} & \cdots & a_{n-1}^{n-1} \end{bmatrix}\!.
\end{equation*}
Both $M x$ and $M^{-1} x$ can be computed in deterministic time $\Order(n \log^2 n)$ using $\Order(n \log^2 n)$ field operations.
\end{theorem}

\noindent
We remark that the transposed Vandermonde matrix $M$ is nonsingular if and only if the elements $a_0, \dots, a_{n-1}$ are pairwise distinct. The next lemma reinterprets this result in terms of multi-point evaluation and interpolation of sparse polynomials. In analogy to the vector notation, we denote by $\supp(A)$ the set of exponents $i$ for which $X^i$ has a nonzero coefficient in $A$, and we say that $A$ is \emph{$t$-sparse} if $|\supp(A)| \leq t$.

\begin{lemma}[Sparse Evaluation and Interpolation] \label{lem:evaluation-interpolation}
Let $\Field$ be a field and let $\omega \in \Field$ have multiplicative order at least $n$. The following two computational problems can be solved in deterministic time $\Order(t \log^2 t + t \log n)$:
\begin{enumerate}
\itemdesc{Evaluation:} Given a $t$-sparse degree-$n$ polynomial $A$, evaluate $A(\omega^0), \dots, A(\omega^{t-1})$.
\itemdesc{Interpolation:} Given $a_0, \dots, a_{t-1} \in \Field$ and a size-$t$ set $T \subseteq [n]$, interpolate the unique polynomial $A$ with evaluations $A(\omega^i) = a_i$ for all $i \in [t]$ and $\supp(A) \subseteq T$.
\end{enumerate}
\end{lemma}
\begin{proof}
\begin{enumerate}
\itemdesc{Evaluation:} Assume that $A$ has the form $A(X) = \sum_{i=1}^t A_{x_i} X^{x_i}$. We precompute the powers $\omega^{x_1}, \dots, \omega^{x_t}$ by repeated squaring in time $\Order(t \log n)$. We can then compute the evaluations $A(\omega^0), \dots, A(\omega^{t-1})$ by computing the following transposed Vandermonde matrix-vector product:
\begin{equation*}
    \begin{bmatrix}
        A(\omega^0) \\ A(\omega^1) \\ \vdots \\ A(\omega^{t-1})
    \end{bmatrix}
    =
    \begin{bmatrix}
        1 & 1 & \cdots & 1 \\
        \omega^{x_1} & \omega^{x_2} & \cdots & \omega^{x_t} \\
        \vdots & \vdots & \ddots & \vdots \\
        \omega^{x_1 (t-1)} & \omega^{x_2 (t-1)} & \cdots & \omega^{x_t (t-1)}
    \end{bmatrix}
    \begin{bmatrix}
        A_{x_1} \\ A_{x_2} \\ \vdots \\ A_{x_t}
    \end{bmatrix}.
\end{equation*}
Since $\omega$ has order at least $n$, the elements $\omega^{x_1}, \dots, \omega^{x_t}$ are pairwise distinct. Therefore, this matrix is nonsingular and we may apply \cref{thm:vandermonde} to efficiently evaluate the product in time $\Order(t \log^2 t)$.

\itemdesc{Interpolation:}
Let $x_1, \dots, x_t$ denote the elements in $T$. We similarly prepare $\omega^{x_1}, \dots, \omega^{x_t}$ via repeated squaring. To interpolate $A$, we solve the following transposed Vandermonde equation system with indeterminates $A_{x_1}, \dots, A_{x_t}$:
\begin{equation*}
\begin{bmatrix}
    a_0 \\ a_1 \\ \vdots \\ a_{t-1}
\end{bmatrix}
=
\begin{bmatrix}
    1 & 1 & \cdots & 1 \\
    \omega^{x_1} & \omega^{x_2} & \cdots & \omega^{x_t} \\
    \vdots & \vdots & \ddots & \vdots \\
    \omega^{x_1 (t-1)} & \omega^{x_2 (t-1)} & \cdots & \omega^{x_t (t-1)}
\end{bmatrix}
\begin{bmatrix}
    A_{x_1} \\ A_{x_2} \\ \vdots \\ A_{x_t}
\end{bmatrix}.
\end{equation*}
Again, this matrix is nonsingular and thus \cref{thm:vandermonde} applies to compute a solution in time $\Order(t \log^2 t)$. Setting $A = \sum_{i=1}^t A_{x_i} X^{x_i}$, we have clearly reconstructed a polynomial with the correct evaluations $A(\omega^i) = a_i$ and support set $\supp(A) \subseteq T$. Moreover, $A$ is the only polynomial satisfying these conditions since the equation system is nonsingular and therefore $A$ is uniquely determined. \qedhere
\end{enumerate}
\end{proof}

\noindent
Using \cref{lem:evaluation-interpolation} we obtain our main result assuming we know a superset of the support and an appropriate element $\omega$.

\begin{lemma}[Sparse Convolution over a Large Field] \label{lem:field-sparse-conv}
Let $\Field$ be a field. Given $A, B \in \Field^n$, a set $T \supseteq \supp(A \conv B)$ and an element $\omega \in \Field$ with multiplicative order at least $n$, we can compute $A \conv B$ in deterministic time $\Order(t \log^2 t + t \log n)$ using $\Order(t \log^2 t + t \log n)$ field operations. Here, $t = \norm A_0 + \norm B_0 + |T|$.
\end{lemma}
\begin{proof}
We follow an evaluation--interpolation approach. Let us identify vectors with polynomials via $A(X) = \sum_{i=0}^{n-1} A_i X^i$. In this correspondence, taking convolutions $A \conv B$ corresponds to multiplying polynomials $A(X) \cdot B(X)$.

We first evaluate $A(\omega^0), \dots, A(\omega^{t-1})$ and $B(\omega^0), \dots, B(\omega^{t-1})$ using \cref{lem:evaluation-interpolation}. We then apply \cref{lem:evaluation-interpolation} again to interpolate a polynomial $C(X)$ with $\supp(C) \subseteq T$ and $C(\omega^i) = A(\omega^i) \cdot B(\omega^i)$ for all $i \in [t]$. One solution is the correct polynomial $C(X) = A(X) \cdot B(X)$, and \cref{lem:evaluation-interpolation} guarantees that this is the unique solution. The running time is $\Order(t \log^2 t + t \log n)$ as claimed.
\end{proof}

\noindent
It remains to construct $\omega$ (see \cref{sec:deterministic:sec:field}) and to find a superset of the support (see \cref{sec:deterministic:sec:support}).

\subsection{Finding Large-Order Elements} \label{sec:deterministic:sec:field}
We next solve the sparse convolution problem for integer vectors $A, B$ assuming that we know the support of $A \conv B$, using what we have established in the last section. We start with the following two lemmas due to Cheng~\cite{Cheng07}; for completeness we include short proofs.

\begin{lemma}[{{{\cite{Cheng07}}}}] \label{lem:irreducible}
Let $\beta \in \Field_p$ be primitive. Then $X^{p-1} - \beta \in \Field_p[X]$ is irreducible.
\end{lemma}
\begin{proof}
Let $f = X^{p-1} - \beta$ and let $f = f_1 \dots f_m$ denote its factorization into monic irreducibles. We first prove that all factors have the same degree. Let $\alpha$ be a root of $f$ (in a field extension). Then $\{ x \alpha : x \in \Field_p^\times \}$ must be the full set of roots of~$f$. Indeed, $(x\alpha)^{p-1} - \beta = x^{p-1} \beta - \beta = 0$ by Fermat's Little Theorem, and there cannot be other roots since~$f$ has degree $p-1$. Pick arbitrary distinct indices $1 \leq i, j \leq m$; we prove that $\deg(f_i) \leq \deg(f_j)$. Let $x, y \in \Field_p^\times$ be such that $x \alpha$ is a root of $f_i$ and $y \alpha$ is a root of $f_j$. We can construct a polynomial $f_j'(X) = f_j(y x^{-1} X)$, which by construction has degree $\deg(f_j)$ and has~$x \alpha$ as a root. But recall that $f_i$ is irreducible (and monic) and therefore the minimal polynomial of~$x \alpha$. It follows that $\deg(f_i) \leq \deg(f_j') = \deg(f_j)$. Since $i, j$ were arbitrary we conclude that all polynomials $f_1, \dots, f_m$ must have common degree \raisebox{0pt}[0pt][0pt]{$d = \frac{p-1}m$}.

Next, we prove that $m = 1$. Let $\alpha_1, \dots, \alpha_d$ denote the roots of $f_1$ (in a field extension). As observed before, we have that $\alpha_i \alpha_j^{-1} \in \Field_p$ for all $i, j$. Moreover, $\prod_i \alpha_i$ is the constant coefficient of $f_1$ and thus $\prod_i \alpha_i \in \Field_p$. It follows that $\alpha_1^d = \prod_{i=1}^d \alpha_1 \alpha_i^{-1} \alpha_i$ is an element of~$\Field_p$. Recall that $\alpha_1$ is a root of $f$ and hence \raisebox{0pt}[0pt][0pt]{$\alpha_1^{p-1} = (\alpha_1^d)^m = \beta$}. Finally, any value $m > 1$ would contradict the primitivity of~$\beta$.
\end{proof}

\begin{lemma}[{{{\cite{Cheng07}}}}] \label{lem:large-order}
Let $f = X^{p-1} - \beta \in \Field_p[X]$ be an irreducible polynomial. Then $X + 1$ has multiplicative order at least $2^p$ in $\Field_p[X] / \ideal f$ provided that $p \geq 7$.
\end{lemma}
\begin{proof}
Let $\Field_{p^{p-1}}$ denote the field $\Field_p[X] / \ideal f$. Let $s$ denote the order of $X + 1 \in \Field_{p^{p-1}}$ and let $S$ denote the set of $s$-th roots of unity in $\Field_{p^{p-1}}$ (that is, $S$ is the set of all polynomials $g \in \Field_p[X] / \ideal f$ such that $g^s = 1 \mod f$). We show that $S$ must be large. We clearly have $X + 1 \in S$. More generally, for any $i \in \Field_p^\times$ we also have $i X + 1 \in S$ since~$X + 1$ and $i X + 1$ are conjugate over $\Field_p$. Furthermore, $S$ is closed under multiplication.

Let $E \subseteq \Nat^{p-1}$ be the set of all sequences $e = (e_1, \dots, e_{p-1})$ with entry sum $\sum_i e_i = p-2$. For any such sequence $e \in E$, we define $\phi(e) = \prod_{i=1}^{p-1} (i X + 1)^{e_i} \in \Field_{p^{p-1}}$. By the previous paragraph, $\phi$ is a map $\phi : E \to S$. We claim that $\phi$ is injective. If $\phi(e) = \phi(e')$ for distinct $e, e' \in E$, then by definition
\begin{equation*}
    \prod_{i=1}^{p-1} (i X + 1)^{e_i} = \prod_{i=1}^{p-1} (i X + 1)^{e_i'} \mod f.
\end{equation*}
Recall that $f$ has degree $p-1$, but $\sum_i e_i = \sum_i e_i' < p-1$. It follows that the equation remains true even without computing modulo $f$:
\begin{equation*}
    \prod_{i=1}^{p-1} (i X + 1)^{e_i} = \prod_{i=1}^{p-1} (i X + 1)^{e_i'}.
\end{equation*}
However, this identity contradicts unique factorization in $\Field_p[X]$. It follows that $\phi$ is injective and therefore $s \geq |S| \geq |E|$. Finally, by a simple counting argument one can show that $|E| = \binom{2p-4}{p-2} \geq 2^p$, for all $p \geq 7$.
\end{proof}

\noindent
For the rest of this section, we will analyze \cref{alg:integer-sparse-conv}.

\begin{lemma}[Correctness of \cref{alg:integer-sparse-conv}] \label{lem:integer-sparse-conv-correct}
Given integer vectors $A, B$ and an arbitrary set $T \supseteq \supp(A \conv B)$, \cref{alg:integer-sparse-conv} correctly returns $C = A \conv B$.
\end{lemma}
\begin{proof}
First, focus on an arbitrary iteration $i$ of the loop in \crefrange{alg:integer-sparse-conv:line:loop}{alg:integer-sparse-conv:line:convolve}. We prove that the algorithm computes the vector $C^i \in \Field_{p_i}$ which is obtained from $C = A \conv B$ by reducing all coefficients modulo~$p_i$. The polynomial $X^{p_i-1} - \beta$ computed in \cref{alg:integer-sparse-conv:line:primitive,alg:integer-sparse-conv:line:field} is indeed irreducible by \cref{lem:irreducible}, so we can represent $\Field_{q_i}$ as $\Field_{p_i} / \ideal{X^{p_i-1} - \beta}$ as claimed. Moreover, the element $\omega \in \Field_{q_i}$ constructed in \cref{alg:integer-sparse-conv:line:omega} has multiplicative order at least $2^{p_i} \geq n$ by \cref{lem:large-order}. The preconditions of \cref{lem:field-sparse-conv} are satisfied ($T \supseteq \supp(A \conv B) \supseteq \supp(A^i \conv B^i)$ and $\omega$ has order at least $n$), hence we correctly compute $C^i = A^i \conv B^i$ in \cref{alg:integer-sparse-conv:line:convolve}. Note that although we carry out the computations over the extension field $\Field_{q_i}$, the vector $C^i$ is guaranteed to have coefficients in $\Field_{p_i}$.

We finally use the Chinese Remainder Theorem to recover $C$ from its images modulo $p_1, \dots, p_k$. As $\prod_{i=1}^k p_i \geq 2^k \geq n \norm A_\infty \norm B_\infty$ exceeds the maximum coefficient in $C$, this recovery step correctly identifies $C = A \conv B$.
\end{proof}

\begin{lemma}[Running Time of \cref{alg:integer-sparse-conv}] \label{lem:integer-sparse-conv-time}
The running time of \cref{alg:integer-sparse-conv} is bounded by $\Order(t \log^4 n \polyloglog n)$ where $t = \norm A_0 + \norm B_0 + |T|$, assuming that $\norm A_\infty, \norm B_\infty \leq \poly(n)$.
\end{lemma}
\begin{proof}
Assuming that $\norm A_\infty, \norm B_\infty \leq \poly(n)$, we have $k = \ceil{\log(n \norm A_\infty \norm B_\infty)} \leq \Order(\log n)$. Note that $p_1, \dots, p_k \leq \widetilde\Order(\log n)$ by the Prime Number Theorem, and therefore computing these primes in \cref{alg:integer-sparse-conv:line:sieve} takes time $\widetilde\Order(\log n)$, using for instance Eratosthenes' sieve. Finding a primitive element $\beta \in \Field_{p_i}$ in \cref{alg:integer-sparse-conv:line:primitive} takes time $\widetilde\Order(\log n)$ as well and \crefrange{alg:integer-sparse-conv:line:field}{alg:integer-sparse-conv:line:fold} have negligible costs. Per iteration, running the convolution algorithm in \cref{alg:integer-sparse-conv:line:convolve} takes time $\Order(t \log^2 t + t \log n) = \Order(t \log^2 n)$ and requires the computation of at most $\Order(t \log^2 n)$ field operations in $\Field_{q_i}$, thus amounting for time $\Order(t \log^3 n \polyloglog n)$. In total the loop in \cref{alg:integer-sparse-conv:line:loop} takes time $\Order(t \log^4 n \polyloglog n)$. Finally, each call to the algorithmic Chinese Remainder Theorem in \cref{alg:integer-sparse-conv:line:recover} takes time $\Order(\log^2(\prod_{i=1}^k p_i)) = \Order(\log^2 n \polyloglog n)$~\cite{vonzurGathenG13}.
\end{proof}

\begin{algorithm}[t]
\caption{} \label{alg:integer-sparse-conv}
\begin{algorithmic}[1]
\Input{Vectors $A, B \in \Int^n$ and a set $T \supseteq \supp(A \conv B)$}
\Output{$C = A \conv B$}
\smallskip
\State Let $k = \ceil{\log(n \norm A_\infty \norm B_\infty)}$
\State Compute the smallest $k$ primes $p_1, \dots, p_k$ larger than $\ceil{\log n}$ \label{alg:integer-sparse-conv:line:sieve}
\For{$i \gets 1, \dots, k$} \label{alg:integer-sparse-conv:line:loop}
    \State Find a primitive element $\beta \in \Field_{p_i}$ by brute-force \label{alg:integer-sparse-conv:line:primitive}
    \State Let $q_i = p_i^{p_i - 1}$ and represent $\Field_{q_i}$ as $\Field_{p_i}[X] / \ideal{X^{p_i-1} - \beta}$ \label{alg:integer-sparse-conv:line:field}
    \State Let $\omega = X + 1 \in \Field_{q_i}$ \label{alg:integer-sparse-conv:line:omega}
    \State Reduce the coefficients of $A, B$ modulo $p_i$ to obtain $A^i, B^i \in \Field_{p_i}^n \subseteq \Field_{q_i}^n$ \label{alg:integer-sparse-conv:line:fold}
    \State Compute $C^i \gets A^i \conv B^i$ over $\Field_{q_i}$ using \cref{lem:field-sparse-conv} with $T$ and $\omega$ \label{alg:integer-sparse-conv:line:convolve}
\EndFor
\ForEach{$x \in T$}
    \State Use Chinese Remaindering to recover $C_x \in \Int$ from $C^1_x \in \Field_{p_1}, \dots, C^k_x \in \Field_{p_k}$ \label{alg:integer-sparse-conv:line:recover}
\EndForEach
\State\Return $C$ with entries $C_x$ for $x \in T$ and zeros elsewhere
\end{algorithmic}
\end{algorithm}

\noindent
We remark that our algorithm can be somewhat simplified by exploiting the following result: For any finite field $\Field_{p^m}$, one can construct in time $\poly(p, m)$ a (simple-structured) set which is guaranteed to contain a primitive element~\cite{Shoup90,Shparlinski92}. The drawback is that the running time worsens by a couple of log factors.

\subsection{Recursively Computing the Support} \label{sec:deterministic:sec:support}
We finally remove the assumption that the support of $A \conv B$ is given as part of the input.

\begin{lemma}[Deterministic Sparse Nonnegative Convolution] \label{lem:deterministic}
There is a deterministic algorithm to compute the convolution of two nonnegative vectors $A, B \in \Nat^n$ in time $\Order(t \log^5 n \polyloglog n)$ where $t = \norm{A \conv B}_0$, assuming that $\norm A_\infty, \norm B_\infty \leq \poly(n)$.
\end{lemma}
\begin{proof}
We construct a recursive algorithm for computing $C = A \conv B$ using the scaling trick from~\cite{BringmannN21}. In order to apply \cref{alg:integer-sparse-conv}, we first recursively compute a set $T \supseteq \supp(C)$. To this end, let~$A', B'$ be vectors of length $\ceil{\frac n2}$ defined by $A'_i = A_i + A_{i+\ceil{n/2}}$ and~$B'_j = B_j + B_{j + \ceil{n/2}}$ (that is, we fold $A$ and $B$ in half). We call the convolution algorithm recursively to compute $C' = A' \conv B'$, and assign
\begin{equation*}
    T = \big\{\, k', k' + \ceil{\tfrac n2}, k' + 2\cdot \ceil{\tfrac n2} : k' \in \supp(C') \,\big\}.
\end{equation*}
We claim that indeed $T \supseteq \supp(C)$. To see this, let $k \in \supp(C)$ and write $k = i + j$ for some $i \in \supp(A)$ and $j \in \supp(B)$. By construction we have $i' = i \bmod \ceil{\frac n2} \in \supp(A')$ and~$j' = j \bmod \ceil{\frac n2} \in \supp(B')$, and thus $k' = i' + j' \in \supp(C')$. By definition it is immediate that $k' \in \{k, k - \ceil{\frac n2}, k - 2\cdot\ceil{\frac n2}\}$ and therefore $k \in T$.

We finally analyze the running time. It takes time $\Order(|T| \log^4 n \polyloglog n)$ to call \cref{alg:integer-sparse-conv} once, by \cref{lem:integer-sparse-conv-correct}. Note that $|T| \leq 3 |\supp(C')| \leq 3t$, and thus each call takes time $\Order(t \log^4 n \polyloglog n)$. Since the length of all vectors is halved in every step, after $\log n$ recursion levels the problem has reached constant input size. The total running time is bounded by $\Order(t \log^5 n \polyloglog n)$.
\end{proof}

\noindent
The proof of \cref{thm:deterministic} is now immediate from \cref{lem:deterministic}. To analyze the algorithm for vectors with entries of size $\Delta$, simply view the vectors as having length $n' = \max\{n, \Delta\}$.

%% file: sections/las-vegas.tex
% !TEX root = ../paper.tex
\section{Las Vegas Algorithms} \label{sec:las-vegas}
The goal of this section is to prove \cref{thm:las-vegas,thm:fast-las-vegas}. We first gather some facts about hash functions (\cref{sec:las-vegas:sec:hashing}) and prove the sparsity testing lemma (\cref{sec:las-vegas:sec:sparsity}). Then we prove \cref{thm:las-vegas} (\cref{sec:las-vegas:sec:simple}) and \cref{thm:fast-las-vegas} (\cref{sec:las-vegas:sec:fast,sec:las-vegas:sec:length-reduction}).

\subsection{Hashing} \label{sec:las-vegas:sec:hashing}
Recall that a linear hash function $h : [n] \to [m]$ is defined by $h(x) = (a x \bmod N) \bmod m$, where $N \geq n$ is fixed and $a \in [N]$ is a random number. Typically $N$ is a prime number, in which case it is easy to prove that the family is $2$-universal. However, this choice is inefficient since precomputing a prime number~$N$ requires time $\polylog N$. In many applications this overhead is negligible---in our case it would incur an additive $\polylog n$ term to the running time which is otherwise independent of $n$. One can remove this overhead and perform linear hashing \emph{without} prime numbers as proven e.g.\ in~\cite{Dietzfelbinger18}. In \cref{sec:hashing} we provide a different self-contained proof.

\begin{restatable}[Linear Hashing without Primes]{lemma}{lemmahashing} \label{lem:hashing}
Let $n \geq m$ be arbitrary. There is a family of linear hash functions $h : [n] \to [m]$ with the following three properties. 
\begin{enumerate}
\itemdesc{Efficiency:} Sampling and evaluating $h$ takes constant time.
\itemdesc{Uniform Differences:} For any distinct keys $x, y \in [n]$ and for any $q \in [m]$, the probability that $h(x) - h(y) \equiv q \mod m$ is at most $\Order(\frac1{m})$.
\itemdesc{Almost-Additiveness:} There exists a constant-size set $\Phi \subseteq [m]$ such that for all keys $x, y \in [n]$ it holds that $h(x) + h(y) \equiv h(x + y) + \phi \mod m$ for some $\phi \in \Phi$.
\end{enumerate}
\end{restatable}

\noindent
To prove \cref{thm:fast-las-vegas} we additionally make use of another family of hash functions: For a random prime number $p$, the hash function $h(x) = x \bmod p$ satisfies similar properties.

\begin{lemma}[Random Prime Hashing] \label{lem:hashing-prime}
Let $n \geq m$ be arbitrary. The family of hash functions $h(x) = x \bmod p$ where $p \in [m, 2m]$ is a random prime satisfies the following three properties.
\begin{enumerate}
\itemdesc{Efficiency:} Sampling $h$ takes time $\polylog m$ and evaluating takes constant time.
\itemdesc{Almost-Universality:} For any distinct keys $x, y \in [n]$, the probability that $h(x) = h(y)$ is at most $\Order(\frac{\log n}{m})$.
\itemdesc{Additiveness:} For all keys $x, y \in [n]$ it holds that $h(x) + h(y) \equiv h(x + y) \mod p$.
\end{enumerate}
\end{lemma}

\subsection{Derivatives and Sparsity Testing} \label{sec:las-vegas:sec:sparsity}
Recall that we define the \emph{derivative} $\partial A$ of a vector $A$ coordinate-wise by $(\partial A)_i = i \cdot A_i$, and we define the \emph{$d$-th derivative} $\partial^d A$ by $(\partial^d A)_i = i^d \cdot A_i$. The crucial ingredient for the Las Vegas guarantee is the following lemma about testing $1$-sparsity of a vector, having access to its first and second derivatives.
\lemmaonesparsity*
\begin{proof}
Note that since $V$ is nonnegative, we can rewrite $V_i = \sqrt{V_i} \cdot \sqrt{V_i}$. The proof is a straightforward application of the Cauchy-Schwartz inequality:
\begin{equation*}
    \norm{\partial V}_1^2 = \left(\sum_i i V_i\right)^2 = \left(\sum_i \sqrt{V_i} \cdot i \sqrt{V_i}\right)^2 \leq \left(\sum_i V_i\right)\left(\sum_i i^2 V_i\right) = \norm V_1 \cdot \norm{\partial^2 V}_1.
\end{equation*}
Recall that the Cauchy-Schwartz inequality is tight if and only if the involved vectors $V$ and~$\partial^2 V$ are scalar multiples of each other. This is possible if and only if $\norm V_0 \leq 1$.
\end{proof}

\subsection{Simple Algorithm} \label{sec:las-vegas:sec:simple}
We are finally ready to analyze \cref{alg:las-vegas}. For the ease of presentation, we have extracted the core part of \cref{alg:las-vegas} (\crefrange{alg:las-vegas:line:hash}{alg:las-vegas:line:update}) as \cref{alg:linear-hashing}, and our first goal is a detailed analysis of that core part.

\begin{algorithm}[t]
\caption{} \label{alg:linear-hashing}
\begin{algorithmic}[1]
\Input{Nonnegative vectors $A, B \in \Nat^n$ and a parameter $m$}
\Output{A nonnegative vector $R \leq A \conv B$, for details see \cref{lem:correctness-linear-hashing}}
\smallskip
\State Sample a linear hash function $h : [n] \to [m]$
\State Compute $X \gets h(A) \conv_m h(B)$ \label{alg:linear-hashing:line:X}
\State Compute $Y \gets h(\partial A) \conv_m h(B) + h(A) \conv_m h(\partial B)$ \label{alg:linear-hashing:line:Y}
\State Compute $Z \gets h(\partial^2 A) \conv_m h(B) + 2 h(\partial A) \conv_m h(\partial B) + h(A) \conv_m h(\partial^2 B)$ \label{alg:linear-hashing:line:Z}
\State Initialize $R \gets (0, \dots, 0)$
\ForEach{$k \in [m]$} \label{alg:linear-hashing:line:loop}
    \If{$X_k \neq 0$ \AND{} $Y_k^2 = X_k \cdot Z_k$} \label{alg:linear-hashing:line:condition}
        \State $z \gets Y_k / X_k$ \label{alg:linear-hashing:line:position}
        \State $R_z \gets R_z + X_k$ \label{alg:linear-hashing:line:update}
    \EndIf
\EndForEach
\State\Return $R$
\end{algorithmic}
\end{algorithm}

To increase clarity we shall adopt the following naming convention for the rest of this section: The indices $x, y, z \in [n]$ exclusively denote coordinates of large vectors, whereas $i, j, k \in [m]$ denote coordinates of the hashed vectors, or equivalently, buckets of a hash function $h$. The first lemma analyzes the vectors $X, Y, Z$ computed by the algorithm.

\begin{lemma}
Let $h, X, Y, Z$ be as in \cref{alg:linear-hashing}. Moreover, for a bucket $k \in [m]$ define the nonnegative vector $V^k \in \Nat^n$ by
\begin{equation*}
    V^k_z = \sum_{\substack{x + y = z\\h(x) + h(y) \equiv k \bmod m}} A_x \cdot B_y.
\end{equation*}
Then $X_k = \norm{V^k}_1$, $Y_k = \norm{\partial V^k}_1$ and $Z_k = \norm{\partial^2 V^k}_1$.
\end{lemma}

\noindent
Note that $A \conv B = \sum_k V^k$. Intuitively, the vector $V^k$ is that part of $A \conv B$ which is hashed into the $k$-th bucket.

\begin{proof}
We merely showcase that $Y_k = \norm{\partial V^k}_1$; the other proofs are very similar. For convenience, let us denote equality modulo $m$ by $\equiv$. It holds that:
\begin{align*}
    Y_k
    &= \big(\, h(\partial A) \conv_m h(B) + h(A) \conv_m h(\partial B) \,\big)_k \\
    &= \sum_{i + j \equiv k} h(\partial A)_i \cdot h(B)_j + h(A)_i \cdot h(\partial B)_j \\
    &= \sum_{\substack{x, y\\h(x) + h(y) \equiv k}} (\partial A)_x \cdot B_y + A_x \cdot (\partial B)_y \\
    &= \sum_{\substack{x, y\\h(x) + h(y) \equiv k}} (x + y) \cdot A_x \cdot B_y \\
    &= \sum_z z \cdot \sum_{\substack{x + y = z\\h(x) + h(y) \equiv k}} A_x \cdot B_y \\
    &= \sum_z z \cdot V^k_z \\
    &= \norm{\partial V^k}_1. \qedhere
\end{align*}
\end{proof}

\noindent
Next, we will prove that in every iteration the algorithm computes a feasible approximation~$R$ to the target vector $A \conv B$.

\begin{lemma}[Correctness and Running Time of \cref{alg:linear-hashing}] \label{lem:correctness-linear-hashing}
Given nonnegative vectors~$A, B$ and any parameter~$m$, \cref{alg:linear-hashing} runs in time $\Order(m \log m)$ and computes a vector $R$ such that for every~$z \in [n]$:
\begin{itemize}
\item $R_z \leq (A \star B)_z$ (always), and
\item $R_z < (A \star B)_z$ with probability at most $c \cdot \norm{A \star B}_0 / m$ for some constant $c$.
\end{itemize}
\end{lemma}
\begin{proof}
Fix an iteration $k \in [m]$ of the loop (\cref{alg:linear-hashing:line:loop}) and suppose that the condition in \cref{alg:linear-hashing:line:condition} is satisfied. Defining $V^k$ as in the previous lemma, we claim that $V^k$ is exactly $1$-sparse. Indeed, on the one hand, $V^k$ is not the all-zeros vector as $\norm{V^k}_1 = X_k > 0$. On the other hand, since $\norm{V^k}_1 \cdot \norm{\partial^2 V^k}_1 = X_k \cdot Z_k = Y_k^2 = \norm{\partial V^k}_1^2$ we have that $\norm{V^k}_0 \leq 1$ by \cref{lem:1-sparse}. Given that~$V^k$ is $1$-sparse, it is easy to check that the value $z := Y_k / X_k$ as computed in \cref{alg:linear-hashing:line:position} is the unique nonzero coordinate in $V^k$, i.e., $\supp(V^k) = \{ z \}$. It follows that the update in \cref{alg:linear-hashing:line:update} is in fact an update of the form ``$R \gets R + V^k$''. Recall that $\sum_k V^k = A \star B$, and thus the first item follows directly.

Next, we focus on the second item. We can assume that $z \in \supp(A \conv B)$ as otherwise the statement is trivial given the previous paragraph. Let $\Phi \subseteq [m]$ be the set from \cref{lem:hashing}. We say that~$z$ \emph{collides} with another index $z'$ if there are $\phi, \phi \in \Phi$ such that $h(z) + \phi \equiv h(z') + \phi' \mod m$. If $z$ does not collide with any other $z' \in \supp(A \conv B)$ then we say that~$z$ is \emph{isolated}. The remaining proof splits into the following two statements:
\begin{itemize}
\item \emph{Each index $z \in \supp(A \conv B)$ is isolated with probability $1 - \Order(\norm{A \conv B}_0 / m)$.} If $z$ collides with another index $z'$ then we have $h(z) - h(z') \equiv q \mod m$ for some $q = \phi - \phi'$, $\phi, \phi' \in \Phi$. For any fixed $q$ this event occurs with probability at most $\Order(\frac1m)$ by the uniform difference property of linear hashing (\cref{lem:hashing}). Taking a union bound over the constant number of elements $q$, we conclude that~$z$ collides with $z'$ with probability at most $\Order(\frac1m)$. Hence the expected number of collisions is $\Order(\norm{A \conv B} / m)$. Using Markov's inequality we finally obtain that a collision occurs with probability at most $\Order(\norm{A \conv B} / m)$ and only in that case $z$ fails to be isolated.
\item \emph{Whenever $z$ is isolated we have $R_z = (A \conv B)_z$.} To see this, it suffices to argue that for all~$k$ of the form $k \equiv h(z) + \phi \mod m$, for some $\phi \in \Phi$, the vectors $V^k$ are at most $1$-sparse. In that case the corresponding iterations $k$ each perform the update ``$R \gets R + V^k$'' in \cref{alg:linear-hashing:line:update} and the claim follows since $R_z = \sum_k V^k_z = (A \conv B)_z$. So suppose that some vector~$V^k$ is at least $2$-sparse, i.e., there exist $x, x' \in \supp(A)$ and~$y, y' \in \supp(B)$ such that $h(x) + h(y) \equiv h(x') + h(y') \equiv k \mod m$ and $x + y \neq x' + y'$. Then either $z' := x + y$ or $z' := x' + y'$ differs from $z$, and we have witnessed a collision between $z$ and $z'$. This contradicts the assumption that~$z$ is isolated.
\end{itemize}
Finally, note that the running time is dominated by the six calls to FFT in \crefrange{alg:linear-hashing:line:X}{alg:linear-hashing:line:Z} taking time $\Order(m \log m)$. The loop (\cref{alg:linear-hashing:line:loop}) only takes linear time.
\end{proof}

\noindent
Recall that \cref{alg:las-vegas} simply calls \cref{alg:linear-hashing} several times and returns the coordinate-wise maximum~$C$ of all computed vectors $R$ as soon as $\norm C_1 = \norm A_1 \cdot \norm B_1$. The bucket size~$m$ increases from iteration to iteration. Given the analysis of \cref{alg:linear-hashing} it remains to prove that \cref{alg:las-vegas} is correct and fast, thereby proving \cref{thm:las-vegas}.

\begin{lemma}[Correctness of \cref{alg:las-vegas}] \label{lem:correctness-las-vegas}
Whenever \cref{alg:las-vegas} outputs a vector $C$, then $C = A \conv B$ (with error probability $0$).
\end{lemma}
\begin{proof}
In \cref{alg:las-vegas:line:update}, $C$ is computed as the coordinate-wise maximum of several vectors~$R$ computed by \cref{alg:linear-hashing}. The previous lemma asserts that $R \leq A \conv B$ (coordinate-wise) and therefore also $C \leq A \conv B$. Moreover, since $C$ was returned by the algorithm we must have $\norm C_1 = \norm A_1 \cdot \norm B_1$ (\cref{alg:las-vegas:line:return}). In conjunction, these facts imply that $C = A \conv B$, since both $C$ and $A \conv B$ are nonnegative vectors.
\end{proof}

\begin{lemma}[Running Time of \cref{alg:las-vegas}] \label{lem:time-las-vegas}
The expected running time of \cref{alg:las-vegas} is $\Order(t \log^2 t)$, where $t = \norm{A \conv B}_0$.
\end{lemma}
\begin{proof}
We first prove that the algorithm terminates with high probability as soon as the outer loop (\cref{alg:las-vegas:line:guess}) reaches a sufficiently large value. More precisely, let $c$ be the constant from \cref{lem:hashing} and fix any iteration of the outer loop with value $m \geq 2ct$. We claim that the algorithm terminates within this iteration with probability at least $1 - t^{-1}$. To this end we analyze the probability of the event~$C_z = (A \conv B)_z$, for any fixed index $z$. Recall that~$C$ is the coordinate-wise maximum of $2 \log m \geq 2 \log t$ vectors $R$ computed by \cref{alg:linear-hashing}. For any such vector $R$, \cref{lem:correctness-linear-hashing} guarantees that $R_z = (A \conv B)_z$ with probability at least $1 - ct / m \geq \frac12$. Hence, the probability that $C_z = (A \conv B)_z$ is at least $1 - 2^{-2\log t} = 1 - t^{-2}$. By a union bound over the $t$ nonzero entries $z$, the probability that algorithm correctly computes $C = A \conv B$ in this iteration is at least $1 - t^{-1}$.

The running time of a single iteration with value $m$ is dominated by the $\Order(\log m)$ calls to \cref{alg:linear-hashing} taking time $\Order(m \log m)$. Sampling the hash functions $h$ has negligible cost (by \cref{lem:hashing}) and so does running the inner-most loop (\cref{alg:las-vegas:line:buckets}). The previous paragraph in particular shows that the algorithm terminates before the $\eta$-th iteration after crossing the critical threshold $m \geq 2ct$, with probability at least $1 - t^{-\eta} \geq 1 - 4^{-\eta}$. Hence, we can bound the expected running time by the total time before this threshold ($m < 2ct$) plus the expected time after ($m = 2^\eta \cdot 2ct$) which can be bounded by a geometric series:
\begin{equation*}
    \sum_{\mu=0}^{\log(2ct)} \Order(2^\mu \log^2(2^\mu)) + \sum_{\eta=0}^\infty 4^{-\eta} \cdot \Order((2^\eta \cdot t) \cdot \log^2(2^\eta \cdot t)) = \Order(t \log^2 t). \qedhere
\end{equation*}
\end{proof}

\noindent
This finished the analysis of \cref{alg:las-vegas}, but not yet the proof of \cref{thm:las-vegas} which additionally claims a tail bound on the running time. To get this additional guarantee, we can modify \cref{alg:las-vegas} to increase $m$ more carefully; see the pseudocode in \cref{alg:las-vegas-high-prob}.

\begin{algorithm}[t]
\caption{} \label{alg:las-vegas-high-prob}
\begin{algorithmic}[1]
\Input{Nonnegative vectors $A, B \in \Nat^n$ and a real paramater $\epsilon > 0$}
\Output{$C = A \conv B$}
\smallskip
\State $C \gets (0, \dots, 0)$
\For{$\mu \gets 0, 1, 2, \dots, \infty$} \label{alg:las-vegas-high-prob:line:mu}
    \For{$\nu \gets 0, 1, 2, \dots, \mu$} \label{alg:las-vegas-high-prob:line:nu}
        \RepeatTimes{$\mu \cdot 2^{\nu / (1 + \epsilon)}$}
            \State Compute $R$ by \cref{alg:linear-hashing} with parameter $m = 2^{\mu-\nu}$
            \State Update $C \gets \max\{C, R\}$ (coordinate-wise) \label{alg:las-vegas-high-prob:line:update}
            \If{$\norm C_1 = \norm A_1 \cdot \norm B_1$} \Return $C$ \label{alg:las-vegas-high-prob:line:return}
            \EndIf
        \EndRepeatTimes
    \EndFor
\EndFor
\end{algorithmic}
\end{algorithm}

\begin{lemma}[Correctness and Running Time of \cref{alg:las-vegas-high-prob}] \label{lem:correctness-las-vegas-high-prob}
Given nonnegative vectors~$A, B \in \Nat^n$ and any parameter $\epsilon > 0$, \cref{alg:las-vegas-high-prob} correctly computes their convolution~$A \conv B$ in expected time $\Order(t \log^2 t)$, where $t = \norm{A \conv B}_0$. Moreover, with probability~$1 - \delta$ it terminates in time
\begin{equation*}
    \Order\!\left(t \log^2(t) \cdot \left(\frac{\log(t / \delta)}{\log t}\right)^{1+\epsilon+\order(1)}\right).
\end{equation*}
\end{lemma}
\begin{proof}
The correctness proof is exactly as in \cref{lem:correctness-las-vegas} and can therefore be omitted. We prove the improved running time bound. Let $c$ be the constant from \cref{lem:correctness-linear-hashing} and focus on the iterations of the outer loops (\cref{alg:las-vegas-high-prob:line:mu,alg:las-vegas-high-prob:line:nu}) with values $\mu = M$ and $\nu = N$, where
\begin{equation*}
    N = \left\lceil (1+\epsilon)\log\left(\frac{\log(t/\delta)}{\log t}\right)\right\rceil~~\text{and}~~M = \ceil{\log(2ct)} + N.
\end{equation*}
In that case we have $m = 2^{\mu-\nu} \geq 2ct$. We claim that the algorithm terminates in this iteration with probability at least $1 - \delta$. To prove this, we again analyze the probability of the event $C_z = (A \conv B)_z$ for any fixed index $z$. The vector $C$ is the coordinate-wise maximum of all vectors $R$ computed in the inner loop (\cref{alg:las-vegas:line:repeat}) and \cref{lem:correctness-linear-hashing} proves that the event $R_z = (A \conv B)_z$ happens with probability at least $1 - c t / m \geq \frac12$. Since the inner loop is repeated $\mu \cdot 2^{\nu/(1 + \epsilon)}$ times, the event $C_z = (A \conv B)_z$ happens with probability at least
\begin{equation*}
    1 - 2^{-\mu \cdot 2^{\nu/(1 + \epsilon)}} \geq 1 - 2^{-\log(t) \cdot \log(t/\delta) / \log(t)} = 1 - \frac{\delta}t.
\end{equation*}
By a union bound over the $t$ nonzero coordinates~$z$, the algorithm computes $C = A \conv B$ (and consequently terminates) with probability at least $1 - \delta$.

Now, to analyze the running time, we have to bound the running time until the algorithm reaches the required values $\mu = M$ and $\nu = N$. The running time of a single execution of the inner-most loop is dominated by the call to \cref{alg:linear-hashing} which takes time $\Order(m \log m)$ by \cref{lem:correctness-linear-hashing}. Thus, with probability $1 - \delta$ the total running time is bounded by
\begin{equation*}
    \sum_{\mu=0}^M \sum_{\nu=0}^\mu \mu \cdot 2^{\nu / (1+\epsilon)} \cdot \Order(2^{\mu-\nu} \log(2^{\mu-\nu})) \leq \Order\!\left(M^2 \sum_{\mu=0}^M 2^\mu \sum_{\nu=0}^\mu 2^{-\epsilon\nu/(1+\epsilon)}\right) \leq \Order(2^M \cdot M^2).
\end{equation*}
Plugging in the definition of $M$ this becomes
\begin{equation*}
    \Order(2^M \cdot M^2) = \Order\!\left(t \cdot \left(\frac{\log(t / \delta)}{\log t}\right)^{1+\epsilon} \cdot M^2\right) = \Order\!\left(t \log^2(t) \cdot \left(\frac{\log(t / \delta)}{\log t}\right)^{1+\epsilon+\order(1)}\right).
\end{equation*}

Finally, we derive from the previous paragraph that expected running time is bounded by~$\Order(t \log^2 t)$. Indeed, the total running time exceeds $\ell \cdot t \log^2 t$ with probability $\ll \Order(\ell^{-3})$, and thus the expected running time is $t \log^2 t \cdot \sum_{\ell=1}^\infty (\ell+1) \cdot \Order(\ell^{-3}) \leq \Order(t \log^2 t)$.
\end{proof}

\noindent
This completes the proof of \cref{thm:las-vegas}; it suffices to plug in some constant $0 < \epsilon < 1$ into \cref{lem:correctness-las-vegas-high-prob} to obtain the claimed running time $\Order(t \log^2 (t / \delta))$.

\subsection{Accelerated Algorithm} \label{sec:las-vegas:sec:fast}
We now speed up \cref{alg:las-vegas-high-prob} in expectation. The crucial subroutine in that algorithm is \cref{alg:linear-hashing} which computes a good approximation $R$ of $A \conv B$. For the improvement we design a similar subroutine which instead computes a good approximation of $A \conv B - C$; see \cref{alg:prime-hashing}. 

\begin{lemma}[Correctness and Running Time of \cref{alg:prime-hashing}] \label{lem:correctness-prime-hashing}
Given vectors \mbox{$A, B, C \in \Int^n$} such that $A \conv B - C$ is nonnegative, and any parameter $m$, \cref{alg:prime-hashing} runs in time $\Order(m \log m)$ and computes a vector $R$ such that for every $z \in [n]$:
\begin{itemize}
\item $R_z \leq (A \conv B - C)_z$ (always), and
\item $R_z < (A \conv B - C)_z$ with probability at most $c \log n \cdot \norm{A \conv B - C}_0 / m$ for some constant~$c$.
\end{itemize}
\end{lemma}
\begin{proof}
Recall that by \cref{lem:hashing-prime} the family of hash functions $h(x) = x \bmod p$ is truly additive, i.e., satisfies $h(x) + h(y) \equiv h(x + y) \mod p$ for all keys~$x, y$. As a consequence, it holds that $X = h(A) \conv_p h(B) - h(C) = h(A \conv B - C)$ and similarly $Y = h(\partial(A \conv B - C))$ and~$Z = h(\partial^2(A \conv B - C))$; the proofs of these statements are straightforward calculations.

The rest of the proof is very similar to \cref{lem:correctness-linear-hashing} and we merely sketch the differences. We analogously define vectors $V^k$ by $V^k_z = (A \conv B - C)_z$ if $z \equiv k \bmod p$ and $V^k_z = 0$ otherwise. Then, by the previous paragraph we have $X_k = \norm{V^k}_1$, $Y_k = \norm{\partial V^k}_1$ and $Z_k = \norm{\partial^2 V^k}_1$. It follows by the same argument, using the sparsity tester (\cref{lem:1-sparse}), that the recovered vector $R$ is exactly $R = \sum_k V^k$, where the sum is over all vectors $V^k$ which are at most $1$-sparse. The first item is immediate since $\sum_{k \in [p]} V^k = A \conv B - C$.

To prove the second item, it suffices to argue that with good probability each nonzero entry~$z$ does not collide with any other nonzero entry $z'$ under $h$. In that case, the vector~$V^k$ for $k = h(z)$ is $1$-sparse and the algorithm correctly computes $R_z = (A \conv B - C)_z$. To see that each index $z$ is likely isolated, we apply the $\Order(\log n)$-universality of $h$ (\cref{lem:hashing-prime}): The probability that~$z$ collides with some fixed index $z'$ is at most $\Order(\log(n) / p) \leq \Order(\log(n) / m)$. Taking a union bound over the $\norm{A \conv B - C}_0$ nonzero entries $z'$ yields the claimed bound.

Finally, observe that the running time is again dominated by the six calls to FFT in \crefrange{alg:prime-hashing:line:X}{alg:prime-hashing:line:Z}, which take time $\Order(m \log m)$. Sampling $h$ takes time $\polylog(m)$ and the loop in \cref{alg:prime-hashing:line:loop} takes linear time.
\end{proof}

\begin{algorithm}[t]
\caption{} \label{alg:prime-hashing}
\begin{algorithmic}[1]
\Input{Vectors $A, B, C \in \Int^n$ such that $A \conv B - C$ is nonnegative, and a parameter $m$}
\Output{A nonnegative vector $R \leq A \conv B - C$, for details see \cref{lem:correctness-prime-hashing}}
\smallskip
\State Sample a random prime $p \in [m, 2m]$ and let $h(x) = x \bmod p$
\State Compute $X \gets h(A) \conv_p h(B) - h(C)$ \label{alg:prime-hashing:line:X}
\State Compute $Y \gets h(\partial A) \conv_p h(B) + h(A) \conv_p h(\partial B) - h(\partial C)$ \label{alg:prime-hashing:line:Y}
\State Compute $Z \gets h(\partial^2 A) \conv_p h(B) + 2 h(\partial A) \conv_p h(\partial B) + h(A) \conv_p h(\partial^2 B) - h(\partial^2 C)$ \label{alg:prime-hashing:line:Z}
\State Initialize $R \gets (0, \dots, 0)$
\ForEach{$k \in [p]$} \label{alg:prime-hashing:line:loop}
    \If{$X_k \neq 0$ \AND{} $Y_k^2 = X_k \cdot Z_k$}
        \State $z \gets Y_k / X_k$
        \State $R_z \gets R_z + X_k$
    \EndIf
\EndForEach
\State\Return $R$
\end{algorithmic}
\end{algorithm}

\begin{algorithm}[t]
\caption{} \label{alg:las-vegas-fast}
\begin{algorithmic}[1]
\Input{Nonnegative vectors $A, B \in \Nat^n$}
\Output{$C = A \conv B$}
\smallskip
\State $C \gets (0, \dots, 0)$
\For{$m \gets 1, 2, 4, \dots, \infty$} \label{alg:las-vegas-fast:line:guess}
    \RepeatTimes{$3 \log \log n$} \label{alg:las-vegas-fast:line:first-repeat}
        \State Compute $R$ by \cref{alg:linear-hashing} with inputs $A, B$ and parameter $m$
        \State Update $C \gets \max\{C, R\}$ (coordinate-wise) \label{alg:las-vegas-fast:line:first-update}
    \EndRepeatTimes
    \RepeatTimes{$2 \log m$} \label{alg:las-vegas-fast:line:second-repeat}
        \State Compute $R$ by \cref{alg:prime-hashing} with inputs $A, B, C$ and parameter $m' = \ceil{\frac{m}{\log n}}$
        \State Update $C \gets C + R$ \label{alg:las-vegas-fast:line:second-update}
    \EndRepeatTimes
    \If{$\norm C_1 = \norm A_1 \cdot \norm B_1$} \Return $C$ \label{alg:las-vegas-fast:line:return}
        \EndIf
\EndFor
\end{algorithmic}
\end{algorithm}

\noindent
Next, to obtain the speed-up over \cref{alg:las-vegas-high-prob}, we combine \cref{alg:linear-hashing,alg:prime-hashing}. The rough idea is that \cref{alg:las-vegas-high-prob} reaches a good (but imperfect) approximation $C$ of $A \conv B$ after only $\log\log n$ iterations of the inner-most loop; after that point $A \conv B - C$ is sufficiently sparse so that a few iterations with \cref{alg:prime-hashing} can correct the remaining errors. The resulting algorithm is summarized in \cref{alg:las-vegas-fast}.

\begin{lemma}[Correctness of \cref{alg:las-vegas-fast}] \label{lem:correctness-las-vegas-fast}
Whenever \cref{alg:las-vegas-fast} outputs a vector $C$, then $C = A \conv B$ (with error probability $0$).
\end{lemma}
\begin{proof}
We first prove that the algorithm maintains the invariant $0 \leq C \leq A \conv B$. There are two types of updates. First, for a vector $R$ computed by \cref{alg:linear-hashing}, the algorithm updates ``$C \gets \max\{C, R\}$''. Since $R$ satisfies $0 \leq R \leq A \conv B$ by \cref{lem:correctness-linear-hashing}, this update maintains the invariant. Second, for a vector $R$ computed by \cref{alg:prime-hashing}, the algorithm update ``$C \gets C + R$''. Since $R$ satisfies $0 \leq R \leq A \conv B - C$ by \cref{lem:correctness-prime-hashing}, this update also upholds the invariant.

It is easy to conclude that the algorithm outputs the correct solution $C = A \conv B$, as this is the only vector $0 \leq C \leq A \conv B$ which also satisfies $\norm C_1 = \norm A_1 \cdot \norm B_1$.
\end{proof}

\begin{lemma}[Running Time of \cref{alg:las-vegas-fast}] \label{lem:time-las-vegas-fast}
The expected running time of \cref{alg:las-vegas-fast} is $\Order(t \log t \log\log n)$, where $t = \norm{A \conv B}_0$.
\end{lemma}
\begin{proof}
For the analysis, we split the execution of the algorithm into two \emph{phases}: The first and initial phase ends as soon as $\norm{A \conv B - C}_0 \leq t / \log^2 n$, and the second phase ends when the algorithm terminates. To analyze the expected running times of both phases, we assume that the outer loop (\cref{alg:las-vegas-fast:line:guess}) has reached a value $m \geq 2ct$, where $c$ is the maximum of the constants in \cref{lem:correctness-linear-hashing,lem:correctness-prime-hashing}. In this case we claim that a single execution of the loop body terminates both phases with probability at least $\frac34$.
\begin{enumerate}
\item For the first phase we analyze the pseudocode in \crefrange{alg:las-vegas-fast:line:first-repeat}{alg:las-vegas-fast:line:first-update}. Fix an arbitrary index $z \in \supp(A \conv B)$. For a vector $R$ computed by \cref{alg:linear-hashing} we have $R_z = (A \conv B)_z$ with probability at least $1 - ct / m \geq \frac12$, by \cref{lem:correctness-linear-hashing}. If any of the vectors $R$ computed in \crefrange{alg:las-vegas-fast:line:first-repeat}{alg:las-vegas-fast:line:first-update} satisfies $R_z = (A \conv B)_z$, then we correctly assign ``$C_z \gets \max\{C_z, R_z\}$'' in \cref{alg:las-vegas-fast:line:first-update} (and we never change that entry for the remaining execution of the algorithm). Since the loop runs for $3 \log \log n$ iterations, the probability that $C_z$ remains incorrect is at most $2^{-3 \log\log n} = (\log n)^{-3}$. Therefore, the expected number of incorrectly assigned coordinates is at most $t / \log^3 n$ and by Markov's inequality that number exceeds $t / \log^2 n$ with probability at most $1 / \log n$. This is less than~$\frac18$ for sufficiently large~$n$.
\item For the second phase we analyze the pseudocode in \crefrange{alg:las-vegas-fast:line:second-repeat}{alg:las-vegas-fast:line:second-update}. Assuming that the first phase is finished, we have $\norm{A \conv B - C}_0 \leq t / \log^2 n$. The argument is similar to the first phase: A vector $R$ computed by \cref{alg:prime-hashing} satisfies $R_z = (A \conv B - C)_z$ with probability at least $1 - c \log n \cdot \norm{A \conv B - C}_0 / m' \geq 1 - ct / m \geq \frac12$. If any vector of the vectors~$R$ computed in \crefrange{alg:las-vegas-fast:line:second-repeat}{alg:las-vegas-fast:line:second-update} satisfies $R_z = (A \conv B - C)_z$ then we correctly update ``$C_z \gets C_z + R_z$'' (and this entry is unchanged for the remaining execution). The probability that $C_z$ is still incorrect after $2 \log m \geq 2 \log t$ iterations is $2^{-2\log t} = t^{-2}$. By a union bound over the $t$ nonzero entries $z$, we have correctly computed $C = A \conv B$ after finishing the loop with probability at least $1 - t^{-1}$. For sufficiently large $t$, this is at least~$\frac78$.
\end{enumerate} 
In combination, with probability $\frac34$ both phases finish and therefore the algorithm terminates within a single iteration of the outer loop. Each iteration takes time $\Order(m \log m \cdot \log\log n)$ (\cref{lem:correctness-linear-hashing}) plus $\Order(m' \log m' \cdot \log m) = \Order(m \log^2 m / \log n)$ (\cref{lem:correctness-prime-hashing}). To bound the total running time, we use that only with probability $4^{-\eta}$ the algorithm continues for another $\eta$ iterations of the outer loop after crossing the critical threshold $m \geq 2ct$. Hence, the expected running time is bounded by $\Order(t \log t \log\log n)$ before that threshold and by
\begin{equation*}
    \sum_{\eta = 1}^\infty 4^{-\eta} \cdot \Order\left((2^\eta \cdot t) \log(2^\eta \cdot t) \log\log n + (2^\eta \cdot t) \frac{\log^2(2^\eta \cdot t)}{\log n}\right) = \Order(t \log t \log\log n)
\end{equation*}
after. In total, the expected time is $\Order(t \log t \log\log n)$ as claimed.
\end{proof}

\subsection{Las Vegas Length Reduction} \label{sec:las-vegas:sec:length-reduction}
As a final step, we can reduce the running time of \cref{lem:time-las-vegas-fast} by replacing the $\log\log n$ factor with $\log\log t$. To this end, we implement a length reduction which reduces the convolution of arbitrary-length vectors to a small number of convolutions of length-$\poly(t)$ vectors. The pseudocode is given in \cref{alg:length-reduction}.

\begin{algorithm}[t]
\caption{} \label{alg:length-reduction}
\begin{algorithmic}[1]
\Input{Nonnegative vectors $A, B \in \Nat^n$}
\Output{$C = A \conv B$}
\smallskip
\RepeatInf
    \State Let $m = \norm A_0^3 \cdot \norm B_0^3$
    \State Sample a linear hash function $h : [n] \to [m]$ \label{alg:length-reduction:line:hash}
    \State Compute $X \gets h(A) \conv_m h(B)$ by \cref{alg:las-vegas-fast}
    \State Compute $Y \gets h(\partial A) \conv_m h(B) + h(A) \conv_m h(\partial B)$ by \cref{alg:las-vegas-fast}
    \State Compute $Z \gets h(\partial^2 A) \conv_m h(B) + 2 h(\partial A) \conv_m h(\partial B) + h(A) \conv_m h(\partial^2 B)$ by Alg.~\ref{alg:las-vegas-fast}
    \State Initialize $C \gets (0, \dots, 0)$
    \ForEach{$k \in [m]$} \label{alg:prime-hashing:line:buckets}
        \If{$X_k \neq 0$ \AND{} $Y_k^2 = X_k \cdot Z_k$} \label{alg:prime-hashing:line:condition}
            \State $z \gets Y_k / X_k$ \label{alg:prime-hashing:line:position}
            \State $C_z \gets C_z + X_k$ \label{alg:prime-hashing:line:update}
        \EndIf
    \EndForEach
    \If{$\norm C_1 = \norm A_1 \cdot \norm B_1$} \Return $C$
    \EndIf
\EndRepeatInf
\end{algorithmic}
\end{algorithm}

\begin{lemma}[Correctness and Running Time of \cref{alg:length-reduction}] \label{lem:correctness-length-reduction}
Given nonnegative vectors~$A, B$, \cref{alg:length-reduction} correctly computes their convolution $A \conv B$. The expected running time is $\Order(t \log t \log\log t)$, where $t = \norm{A \conv B}$.
\end{lemma}
\begin{proof}
We skip the correctness part since the proof is exactly like the correctness argument of \cref{alg:las-vegas}; the only difference here is that $X, Y, Z$ are computed by \cref{alg:las-vegas-fast} instead of FFT, however, \cref{alg:las-vegas-fast} is a Las Vegas algorithm and therefore also always correct.

To analyze the running, we start by lower bounding the probability that any iteration terminates the algorithm. We say that a linear hash function $h$ as sampled in \cref{alg:length-reduction:line:hash} is \emph{good} if for all distinct $z, z' \in \supp(A \conv B)$ and all $\phi, \phi' \in \Phi$ it holds that $h(z) + \phi \not\equiv h(z') + \phi' \mod m$; here $\Phi$ is the set in \cref{lem:hashing}. Following the same arguments as in \cref{sec:las-vegas:sec:simple} one can prove that \cref{alg:length-reduction} terminates as soon as a good hash function is sampled. Therefore, we now lower bound the probability that a random linear hash function $h$ is good. For fixed~$z, z', \phi, \phi'$, the probability that $h(z) + \phi \equiv h(z') + \phi' \mod m$ is at most $\Order(\frac1m)$. We take a union bound over the $\Order(t^2)$ choices of $z, z', \phi, \phi'$ and conclude that a random function~$h$ is good with probability at least $1 - \Order(t^2 / m)$. Observe that $\norm A_0 + \norm B_0 - 1 \leq t \leq \norm A_0 \cdot \norm B_0$, and thus $t^3 \leq m \leq t^6$. Therefore, for sufficiently large $t$ each iteration of the loop terminates the algorithm with probability at least $\frac12$.

The running time of each iteration $i$ is dominated by the six convolutions computed by \cref{alg:las-vegas-fast}. Let $T_{i,1}, \dots, T_{i,6}$ denote the running times of these calls, respectively. Moreover, let $S_i$ denote the random variable which indicates whether the $i$-th iteration takes place (or whether the algorithm has terminated before). By the previous paragraph we have that $\Pr(S_i = 1) \leq 2^{-i}$. The total running time is bounded by
\begin{equation*}
    \sum_{i=1}^\infty S_i \cdot \sum_{j=1}^6 T_{i,j}.
\end{equation*}
Hence, by linearity of expectation and since the random variables $S_i$ and $T_{i,j}$ are independent, the expected running time is at most
\begin{equation*}
    \sum_{i=1}^\infty \Ex(S_i) \cdot \sum_{j=1}^6 \Ex(T_{i,j}) \leq \sum_{i=1}^\infty 2^{-i} \cdot \Order(t \log t \log\log m) = \Order(t \log t \log\log t).
\end{equation*}
Here, we used the expected time bound from \cref{lem:time-las-vegas-fast} to bound $\Ex(T_{i,j})$.
\end{proof}

\noindent
\cref{lem:correctness-length-reduction} completes the proof of \cref{thm:fast-las-vegas}.

%% file: sections/linear-hashing.tex
% !TEX root = ../paper.tex
\section{Linear Hashing without Primes} \label{sec:hashing}
In this section we provide the proof for \cref{lem:hashing}.

\lemmahashing*

\noindent
For the proof we need another result. We say that a set $A = \{ r + ia : i \in [|A|] \} \subseteq \Int$ is an \emph{arithmetic progression} with \emph{step-width} $a$. The following lemma proves that two arithmetic progressions with coprime step-widths are as uncorrelated as possible.

\begin{lemma} \label{lem:coprime-progressions}
Let $A$ and $B$ be arithmetic progressions with coprime step-widths $a$ and $b$, respectively. Then $|A \cap B| \leq \min\{\frac{|A|-1}b, \frac{|B|-1}a\} + 1$.
\end{lemma}
\begin{proof}
We may assume that $A = \{0, a, \dots, (|A| - 1)a\}$ and $B = \{0, b, \dots, (|B| - 1)b\}$ (remove all points before the first common element of $A$ and $B$ and shift such that the first common element becomes zero). Since $a$ and $b$ are coprime, the intersection $A \cap B$ consists only of multiples of $ab$ and thus \raisebox{0pt}[0pt][0pt]{$|A \cap B| \leq \floor{\frac{(|A| - 1)a}{ab}} + 1$}. The same bound holds symmetrically for~$B$.
\end{proof}

\begin{proof}[Proof of \cref{lem:hashing}]
First, assume that $m$ is odd. Let $N$ be the smallest power of two larger than $n \cdot m$. We then define the family of hash functions as $h(x) = (a x \bmod N) \bmod m$, where $a \in [N]$ is a random \emph{odd} number. We will now prove the three claimed properties for this family.
\begin{enumerate}
\itemdesc{Efficiency:} Sampling $h$ only involves constructing $N$ and sampling a random odd number. Both operations take constant time in the Word RAM model. Evaluating $h$ is also in constant time.
\itemdesc[3]{Almost-Additiveness:} Fix any keys $x, y \in [n]$. Then for one of the two choices $\phi \in \{0, N\}$ it holds that $(a x \bmod N) + (a y \bmod N) = (a (x + y) \bmod N) + \phi$. By reducing this equation modulo $m$, it follows that $\Phi = \{0, N \bmod m\}$ is a suitable choice.
\itemdesc[2]{Uniform Differences:} To prove that $h$ satisfies the uniform difference property, it suffices to prove that $h$ is $\Order(1)$-uniform, that is, $\Pr(h(z) = \psi) \leq \Order(\frac1m)$ for all~$z \in [n]$ and~$\psi \in [m]$. Indeed, by the previous paragraph we have $h(x) - h(y) \equiv q \mod m$ only if $h(x - y) \equiv q-\phi \mod m$ for some $\phi \in \Phi$. Taking a union bound over the constant number of elements $\phi$ then yields the claim.

To check that $h$ is $\Order(1)$-uniform, we write $z = 2^k \cdot w$ where $w$ is odd. Then $a z \bmod N$ is uniformly distributed in $A = \{ 2^k \cdot i : i \in [2^{-k} \cdot N] \}$. Indeed, by identifying $[N]$ with the finite ring~$\Int / N \Int$, $A$ is the smallest additive subgroup of $\Int / N \Int$ which contains $z$, and thus multiplying with a random unit $a \in (\Int / N \Int)^\times$ randomly permutes $z$ within that subgroup. It follows that
\begin{equation*}
    \Pr(h(z) = \psi) = \frac{2^k \cdot |A \cap B|}{N},
\end{equation*}
where $B \subseteq [N]$ consists of all numbers equal to $\psi$ modulo $m$. Observe that $A$ and $B$ are both arithmetic progressions with step-widths $2^k$ and $m$, respectively. Recall that $m$ is odd, therefore $2^k$ and $m$ are coprime and \cref{lem:coprime-progressions} applies and yields $|A \cap B| \leq \frac{N}{2^k m} + 1$. We finally obtain $\Pr(h(z) = \psi) \leq \frac1m + \frac nN \leq \Order(\frac1m)$.
\end{enumerate}
Finally, we remove the assumption that $m$ is odd. If $m$ is even, then we simply apply the previous construction for $m-1$ to obtain a linear hash function $h : [n] \to [m-1]$ and reinterpret this as a function $h : [n] \to [m]$. It is easy to see that this preserves efficiency and uniform differences, and we claim that is also preserves almost-additiveness. Indeed, for arbitrary keys~$x, y$ we know that $h(x) + h(y) \equiv h(x + y) + \phi \mod{(m-1)}$ for some~$\phi \in \Phi$. Both sides of the equation are integers less than $2m-2$ and hence their images modulo $m-1$ and $m$, respectively, differ by at most $1$. Therefore, we have $h(x) + h(y) \equiv h(x + y) + \phi' \mod m$ for some $\phi' \in \Phi' = \{ \phi + \sigma : \phi \in \Phi, \sigma \in \{-1, 0, 1\}\}$.
\end{proof}